\documentclass[11pt]{article}
\usepackage{multicol}
\newtheorem{remark}{Remark}
\newtheorem{theorem}{Theorem}

\newtheorem{proof}{Proof}
\usepackage{graphicx}
\usepackage{dcolumn}
\usepackage{bbm}
\usepackage{bm}
\usepackage{geometry}
\usepackage{cite}
\usepackage{braket}
\usepackage{lipsum}
\usepackage{mathtools}
\usepackage{lipsum}
\usepackage{bbold}
\usepackage{subcaption}

\usepackage{xcolor}
\usepackage[colorlinks=true,allcolors=blue]{hyperref}
\newcommand*{\SavedEqref}{}
\let\SavedEqref\eqref
\renewcommand*{\eqref}[1]{%
	\begingroup
	\hypersetup{
		linkcolor=linkequation,
		linkbordercolor=linkequation,
	}%
	\SavedEqref{#1}%
	\endgroup
}
\DeclareMathOperator{\Tr}{Tr}
\usepackage{algorithm}
\usepackage{algpseudocode}

\begin{document}
	\title{A Probabilistic Framework for Controlling Quantum Systems}
	
	\author{ Randa Herzallah\thanks{\textcolor{blue}{herzallah.r@gmail.com}}\hspace{0.2cm} and\hspace{0.1cm}Abdessamad Belfakir\thanks{\textcolor{blue}{abdobelfakir01@gmail.com}}\\Aston Institute for Urban Technologies and the Environment, College of
		Engineering and Applied Science,\\ Aston University, Aston Triangle, Birmingham B4 7ET,
		UK.}
	
	\date{\vspace{-5ex}}
	\maketitle
	\begin{abstract}
A new control method that considers all sources of uncertainty and noises that might affect the time evolutions of quantum physical systems is introduced. Under the proposed approach, the dynamics of quantum systems are characterised by probability density functions (pdfs), thus providing a complete description of their time evolution. Using this probabilistic description, the proposed method suggests the minimisation of the distance between the actual pdf that describes the joint distribution of the time evolution of the quantum system and the external electric field, and the desired pdf that describes the system target outcome. We start by providing the control solution for quantum systems that are characterised by arbitrary pdfs. The obtained solution is then applied to quantum physical  systems characterised by Gaussian pdfs and the form of the optimised controller is elaborated. Finally, the proposed approach is demonstrated on a real molecular system and two spin systems showing the effectiveness and simplicity of the method.  
	\end{abstract}
	{\quotation\noindent{\bf Keywords:}  Probabilistic control, 
		Kullback-Leibler Divergence (KLD), quantum control.
		
		\endquotation}
	
\section{Introduction }\label{sec1}
The idea of controlling quantum physical systems dates back to the end of last century  when laser technology was invented \cite{Brumer,Kawashima,Weiner,Kosloff89,Jortner,Tannorand,Tannor}. Since it lies at the heart of  all quantum information theory and quantum technology \cite{Ramos,Krausz,Silberberg,Rego,Ohmori}, this field has been rapidly growing in recent years both theoretically and experimentally. Theoretically, several approaches to control quantum physical systems were developed including optimal control theory \cite{Werschnik,Zhu2000},  Lyapunov
control approaches \cite{Mirrahimi,Kuang,Hou}, learning control algorithms \cite{DaoyiDong}, and robust control methods \cite{Koswara}. The main objective in optimal quantum control  is to design a control input that renders a  closed-loop quantum system stable and optimizes an associated predefined cost function \cite{Kosloff89}. In other words, the objective is to find a controller that transforms the quantum system from an initial state to a predefined desired one \cite{Kosloff89}. To achieve the control objective, a cost function needs to be optimised which is conventionally taken to be the expected value of a predefined target operator\cite{Kosloff89}. Consequently, several numerical algorithms allowing the optimisation of quantum control systems were introduced such as the Krotov method \cite{Kazakov} and the  the rapid monotonically convergent iteration methods introduced first by   Rabitz and his co-workers in a serie of papers\cite{RBZ98,RBZ98_2,Conju91,Peirce98}. These algorithms were first applied to manipulate the transition probability between the bound states of molecular systems \cite{RBZ98} and then extended to other target operators\cite{RBZ98_2}.

A Feedback control method for quantum physical systems using the  best estimation of the dynamical variables from the measurement record was implemented in \cite{Doherty}. In this scheme the best estimates of quantum system variables are used to fully control the system when the measurement noise is assumed to be Gaussian. With this assumption, the quantum mechanical cost function was shown to have the same linear and quadratic form as the one of classical systems\cite{Doherty}. This rendered the application of  the classical linear quadratic Gaussian control theory to quantum systems\cite{Doherty,Jiang} possible. However, despite of  the success of this method and its important applications, it is not general and is only  applied for particular systems under particular assumptions.  Furthermore, due to their primordial applications in quantum optics and quantum information theory, recent advances in quantum control have also been directed towards the control of spin systems \cite{spin,spin_2}. In \cite{spin_2}, the radio-frequency pulses were implemented in coherent spectroscopy in order to control the transfer between the states of multiple spin systems. Moreover, a method to control an ensemble of spin systems that are uniformly distributed was introduced and it was shown that the approach can be extended to Gaussian distributions \cite{spin}. 

However, almost all of the aforementioned methods considered a deterministic cost function to optimise the system and design a suitable control strategy. Thus, these methods do not take into consideration the effect of  noises and uncertainties, that are inherent to real world quantum physical systems, in the design of the control strategy. Consequently, this paper presents a fundamentally different approach to current state of the art of quantum control, that provides a robust control method by considering the different sources of uncertainties affecting  quantum physical systems operating under real noisy and uncertain conditions. The main novelty of our approach is that it defines the cost function of the quantum control problem as the distance between the probabilistic description of the joint distribution of the quantum system dynamics and external electric field, and a predefined desired joint distribution that characterises the desired outcome of the controlled quantum system. Within the proposed method, the time evolution of quantum systems are characterised by probability density functions (pdfs), hence yielding a complete characterisation of the system evolution under noisy and uncertain conditions. This overcomes many of the drawbacks of the existing quantum control methods and allows the consideration of the quantum system uncertain information in the design of the optimal control strategy. The use of the distance measure between pdfs to optimise optimal control laws is a well established method in classical control theory \cite{Karny_1,RH_2011,RH_2013,RH_2015,RH_2018,RH_2020}. This method will be extended here to control atomic scale systems.

To reemphasise, the advantages of the approach presented in this paper are as follows: Firstly, the designed optimal control law obtained from this method takes into consideration all sources of noises and uncertainties such as sensors noises and measurement uncertainties which may affect the evolution of real world  quantum physical systems. Secondly, in contrast to existing methods in the field, a closed form solution for the optimal control law is obtained for quantum systems described by arbitrary pdfs. This yields an analytic control form for quantum systems that are characterised by Gaussian pdfs and whose time evolutions are represented by  a bilinear state space form obtained from  the Liouville-von Neumann equation. Thirdly, in contrast to previously discussed methods in the literature \cite{RBZ98,RBZ98_2} that require a positive definite target operator  to guarantee convergence of the objective function, our proposed method is not restricted by that requirement.

This paper is organised as follows: Section (\ref{Evolution of open quantum systems}) briefly recalls the evolution of open quantum systems and develops the corresponding state space model. In Section (\ref{Fully Probabilistic Control for Quantum Systems}), we introduce a general theory to fully control the quantum systems using a probabilistic approach and demonstrate its general solution. The developed approach will then be applied in Section (\ref{QuantumCA}) to systems described by Gaussian pdfs . Thereafter, the method is applied to a molecular system and spin systems in Section (\ref{Results and discussions}). Finally, some conclusions are provided in Section (\ref{conclu}).
\section{Evolution of open quantum systems}\label{Evolution of open quantum systems}
\subsection{A brief review of Liouville-von Neumann equation}
The time evolution of a  quantum open system can be described by the following Liouville-von Neumann equation,
	\begin{equation}\label{LVN_eq}
	i \hbar\dfrac{d\rho(t)}{d t} =[H_0+H_u(t),\rho(t)]+\mathcal{L}(\rho(t)), \hspace*{0,2cm} \rho(0)=\rho_0,
\end{equation} 
where $\hbar$ is the reduced Planck's constant, $\rho(t)$ is the reduced density operator which is a positive semidefinite Hermitian operator whose trace, $\Tr({\rho(t)})=1$, ${H}_{0}$ is the free Hamiltonian, and  $H_u(t)=-\mu u(t)$ describes the interaction between the system and the external electric field ${u}(t)$ through the operator $\mu$ related to the system \cite{Neumann,Nemes,redu,wavepack1}. The commutator in equation (\ref{LVN_eq}) describes the system  Liouvillian super-operator for the reduced density operator $\rho(t)$ while $\mathcal{L}(\rho(t))$ represents the open system Liouvillian. In the Linblad approach, it is provided by, 
	\begin{equation}
	\mathcal{L}(\rho(t))=i\hbar\sum_{s}\big(L_s\rho(t) L_s^\dagger-\dfrac{1}{2}\{L_s^\dagger L_s,\rho(t)\}\big),
\end{equation}  
where $L_s$ are the Linblad operators defined in terms of the dissipative
transition rates $\Gamma_{k\to j}$ from  the free Hamiltonian  eigenstate $\ket{k}$ to the eigenstate $\ket{j}$ as,
\begin{equation}
	L_s=L_{j,k}=\sqrt{\Gamma_{k\to j}}\ket{j}\bra{k}.
\end{equation}
In this equation, it is assumed that $k$ takes a finite number of values, i.e., $\{\ket{k}, k = 0,..., l-1\}$  with $l$ being the number of the eigenvectors of the free Hamiltonian \cite{redu,wavepack1}. From equation (\ref{LVN_eq}), it can be easily shown that the time evolution of the elements of the density operator  is given by, 
	\begin{equation}\label{rho_elements}
	\dfrac{d\rho_{n,m}(t)}{d t}=(-i\omega_{n,m}-\gamma_{n,m})\rho_{n,m}(t)+\sum_{k=0}^{l-1}\Gamma_{k\to n}\rho_{k,k}(t)\delta_{n,m}+i\dfrac{{u}(t)}{\hbar}\sum_{k=0}^{l-1}(\mu_{n,k}\rho_{k,m}(t)-\rho_{n,k}(t)\mu_{k,m}),
\end{equation}
where $\{n,m\}=\{0,1,\dots,l-1\}$, $\delta_{n,m}$ is the Kronecker delta function, $\mu_{k,n}:=\bra{k}\mu\ket{n}$ are the elements of the operator $\mu$, $\omega_{n,m}:=\dfrac{E_n-E_m}{\hbar}$ are the Bohr frequencies, with $E_n$ being the energy eigenvalue of the free Hamiltonian $H_0$ associated with the eigenvector $\ket{n}$, and $\gamma_{n,m}$ is the total dephasing rate defined by,
\begin{equation}
	\gamma_{n,m}:=\dfrac{1}{2}\sum_{j=0}^{l-1}(\Gamma_{n\to j}+\Gamma_{m\to j}).
\end{equation}
The time evolution of a physical property, described by an Hermitian operator $\hat{o}$, of the quantum open system whose evolution is described in equation (\ref{LVN_eq}) can be computed by using the following relation, 
\begin{equation}\label{obervable_evolution}
\tilde{o}(t)=\braket{\hat{o}}=\Tr(\rho(t)\hat{o}).
\end{equation}

\subsection{Evolution of the vectorisation of the density operator}\label{sec2_2}
The density matrix $\rho(t)$ described in equation (\ref{LVN_eq}) can be written as follows,
 \begin{equation}
\rho(t)=(\rho(t))^\dagger=\left(
	\begin{array}{ccccccccc}
		\rho_{0,0}(t)&\rho_{0,1}(t) & \dots &\rho_{0,l-1}(t) \\
		\rho_{1,0}(t)&\rho_{1,1}(t) & \dots &\rho_{1,l-1}(t) \\
		\vdots&\vdots & \ddots &\vdots \\
		\rho_{l-1,0}(t)&\rho_{l-1,1}(t) & \dots &\rho_{l-1,l-1}(t) \\
	\end{array}
	\right)\in\mathbb{C}^{l \times l}.\end{equation}
By using the following vectorisation of $\rho(t)$,
\begin{align}\label{x_t_vet}
\tilde{x}(t)&=\text{vec}(\rho(t)) \nonumber \\&=\begin{bmatrix}
\rho_{0,0}(t) & \rho_{1,1}(t) 
\dots
\rho_{l-1,l-1}(t)& \rho_{0,1}(t) \dots  
\rho_{0,l-1}(t)&\rho_{1,0}(t) 
\dots\rho_{l-1,0}(t) 
\dots\dots\dots\rho_{l-1,1}(t)\dots\rho_{l-1,l-1}(t) 
\end{bmatrix}^T,
\end{align}
the differential equations (\ref{rho_elements}) can be equivalently written as,
\begin{equation}\label{NLVN_pa}
	\dfrac{d\tilde{x}(t)}{d t} =(\tilde{A}+iu(t) \tilde{N})\tilde{x}(t),\hspace{0.5cm} \tilde{x}(0)=\tilde{x}_0,
\end{equation} 
where $\tilde{A} \in\mathbb{C}^{l^2 \times l^2}$,  $\tilde{N} \in\mathbb{C}^{l^2\times l^2}$ are some matrices whose elements can be found from equation (\ref{rho_elements}), and $\tilde{x}_0$ is the vectorisation of the initial density operator \cite{redu}. In addition $T$ in equation (\ref{x_t_vet}) stands for the transpose operation. An example of how the $\tilde{A}$ and $\tilde{N}$ matrices can be constructed  for a 2-dimensional physical system is given in Appendix (\ref{vect_app}). The demonstration in the appendix can be generalised to $l$-dimensional physical system for $l>2$.

By using the following shifting, 
\begin{equation}
x(t)\to\tilde{x}(t)-x_e,
\end{equation}
the differential equation (\ref{NLVN_pa}) can  be re-written as follows, 
\begin{equation}\label{state_space}
	\dfrac{dx({t})}{d t}=\tilde{A} x({t})+\tilde{B}(x(t)) u(t) ,\hspace{0.5cm} x(0)=\tilde{x}_0-x_e, 
\end{equation}
where $x_e$ is an eigenvector of $\tilde{A}$ with $\tilde{A}x_e=0$,  and $\tilde{B}(x(t)) =i\tilde{N}(x(t)+x_e)$. This equation
 describes the dynamics of the system $x(t)$ depending
on the  input electric field ${u(t)}$ \cite{redu}.\\
Discretising the state space equation (\ref{state_space}), and using  $x_t\equiv{x}(t)$ and $u_t\equiv{u}(t)$, the discrete time state space representation can be written as follows,
\begin{align}\label{StateQua0}
	x_{t+1} = A x_t + B(x_t) u_t,
\end{align}
where,
\begin{align}\label{A_B_t}
	A &= e^{\tilde{A} \Delta t}, \nonumber\\
	B(x_{t})&= \bigg( \int_0^{\Delta t} e^{\tilde{A} \lambda} \tilde{B}(x_\lambda) \mathrm{d} \lambda\bigg), 
\end{align}
and where $\lambda = \Delta t - t$ with $\Delta t$ being the sampling period. Equivalently, we can re-write the relation (\ref{StateQua0}) as,
\begin{align}\label{StateQua}
	x_{t} = A x_{t-1} + B(x_{t-1}) u_{t-1}.
\end{align}
As discussed in the introduction section and according to relevant literature in the field \cite{Doherty},  various sources of uncertainties might affect the evolution of  quantum systems including for example sampling, parameters, and functional uncertainties. Hence, a noise term that accounts for the various sources of uncertainty is added to the discretised equation (\ref{StateQua}) to give,
\begin{equation}\label{Inpu2NoiseDelay}
	{x}_{t}=Ax_{t-1}+B(x_{t-1}) u_{t-1} + \zeta_{t},
\end{equation}
where $\zeta_{t}$ is a multivariate Gaussian noise. 
According to the discretised state vector of the vectorised density operator defined in equation (\ref{Inpu2NoiseDelay}), the measurement equation can be defined as follows,
\begin{equation}\label{outp}
	o_t=Dx_t+\sigma_t,
\end{equation}
where  $D=(\text{vec}(\hat{o}^T))^T $, and $\sigma_t$ is a multivariate Gaussian noise which affects the observed values, $o_t$.  This equation is called the output equation. Equations (\ref{Inpu2NoiseDelay}) and (\ref{outp}) are referred to as the bilinear state space model that characterise quantum systems.
\section{Fully Probabilistic Control for Quantum Systems}\label{Fully Probabilistic Control for Quantum Systems}
In this  section we will state the control objectives and provide
the general solution of the probabilistic control problem for  quantum systems  described by arbitrary probabilistic dynamical models.
\subsection{Objectives of the Fully Probabilistic Control for Quantum Systems}
Due to the presence of  the noise signal, $\zeta_{t}$ in equation (\ref{Inpu2NoiseDelay}), the system state at time $t$ can not be completely specified using the previous state and control input. It can only be completely determined from its probabilistic description,  $s(x_{t} | u_{t-1}, x_{t-1})$.  Similarly, the measurement output can be only determined by a suitable pdf,
 $s(o_{t} | x_{t})$. Following this probabilistic description, the control problem will be stated as follows: Find the pdf of the randomised controller $c(u_{t-1}|x_{t-1})$ that  minimises the Kullback-Leibler divergence (KLD) between the joint pdf of the closed-loop description of the system dynamics, $f(\mathcal{Z}(t, \mathcal{H}))$ and a predefined ideal one, $^If(\mathcal{Z}(t, \mathcal{H}))$,
 \begin{equation}\label{KLD}
 	\mathcal{D}(f||^I f)=\int f(\mathcal{Z}(t, \mathcal{H}))\ln\big(\dfrac{f(\mathcal{Z}(t, \mathcal{H}))}{^If(\mathcal{Z}(t, \mathcal{H}))}\big)d\mathcal{Z}(t, \mathcal{H}),
 \end{equation}
 where $\mathcal{Z}(t, \mathcal{H}) = \{x_t, \dots, x_\mathcal{H}, o_t, \dots, o_\mathcal{H}, u_{t-1}, \dots, u_\mathcal{H}\}$ is the system closed-loop data sequence, and $\mathcal{H} \le \infty$ is a given control horizon. The joint probability function  of the closed-loop description of the system dynamics    can be evaluated as follows,
 \begin{equation}\label{JointDist}
 	f(\mathcal{Z}(t, \mathcal{H})) = \prod_{t=1}^\mathcal{H} s(x_t|x_{t-1},u_{t-1})s(o_t|x_{t})c(u_{t-1}|x_{t-1}).
 \end{equation}
 Similarly, the ideal joint pdf of the closed-loop data can be factorised as follows,
 \begin{equation}\label{IdealJointDist}
 	^If(\mathcal{Z}(t, \mathcal{H}))=\prod_{t=1}^\mathcal{H} {s(x_t|x_{t-1},u_{t-1})} {^Is(o_t|x_{t})}{^Ic(u_{t-1}|x_{t-1})},
 \end{equation}
 where $^Is(o_t|x_{t})$ represents the ideal distribution of the measurement $o_t$, and $^Ic(u_{t-1}|x_{t-1})$ represents the ideal pdf of the controller. Also, note that in equation (\ref{IdealJointDist}), the factor, $s(x_t|x_{t-1},u_{t-1})$ describing the ideal distribution of the state vector of the vectorised density operator, is taken to be equal to the corresponding counterpart in equation (\ref{JointDist}). This indicates that the evolution of the density matrix will evolve according to the Liouville von-Neumann equation.
 
Following \cite{Karny_1,RH_2011,RH_2013}, the derivation of the
 optimal randomised controller that minimises the KLD given in equation (\ref{KLD}) can be obtained recursively by defining  the expected minimum cost-to-go function $-\ln(\gamma(x_{t-1}))$ governed by, 
 \begin{align}\label{OptPer}
 	-\ln(\gamma(x_{t-1})) &=  \underset{c(u_{t-1}|x_{t-1})}{\text{min}}\sum_{\tau=t}^\mathcal{H} \int f(\mathcal{Z}_t, \dots, \mathcal{Z}_\mathcal{H}|x_{t-1}) \ln \left( \dfrac{s(o_\tau|x_\tau)c(u_{\tau-1}|x_{\tau-1})}{^Is(o_\tau|x_\tau)^Ic(u_{\tau-1}|x_{\tau-1})} \right)  d(\mathcal{Z}_t, \dots, \mathcal{Z}_\mathcal{H}),
 \end{align} 
 for arbitrary $\tau \in \{1, . . . , \mathcal{H}\}$.
 \begin{theorem}\label{theo1}	
 Using the definition of the cost-to-go function given in equation (\ref{OptPer}), the minimisation
 of Kullback-Leibler divergence (\ref{KLD}) between the joint pdf (\ref{JointDist}) and the ideal one (\ref{IdealJointDist}) can be performed recursively to
 give the following recurrence functional equation,
 	\begin{align}
 		\label{cost_to_go_2}
 	-\ln(\gamma(x_{t-1})) &=  \underset{c(u_{t-1}|x_{t-1})}{\text{min}} \int\bigg[ s(x_t|x_{t-1}, u_{t-1})s(o_t|x_t)c(u_{t-1}|x_{t-1})\bigg(\ln \left( \dfrac{s(o_t|x_t)c(u_{t-1}|x_{t-1})}{^Is(o_t|x_t)^Ic(u_{t-1}|x_{t-1})} \right) \nonumber \\ &-\ln(\gamma(x_{t}))\bigg)  \bigg]d(x_t,o_t,u_t),
 	\end{align}
 \end{theorem}
 \begin{proof}
 	The derivation of the above result can be obtained by evaluating the optimal cost-to-go function specified in equation (\ref{OptPer}).
 \end{proof}
\subsection{ General solution to  the Fully Probabilistic quntum Control Problem}\label{SFPDCQS}
The general solution of the  control
problem  obtained from the minimisation of the cost-to-go
function with respect to control distribution, $c(u_{t-1}|x_{t-1})$, can be obtained as shown in the following theorem.
\begin{theorem}\label{Prop1}	
	The pdf of the optimal control law, $c(u_{t-1}|x_{t-1})$, that minimises the cost-to-go function (\ref{cost_to_go_2}) can be shown to be given by, 
	\begin{equation}
		\label{eq:Eqn4}
		c(u_{t-1}|x_{t-1})=\frac{^Ic(u_{t-1}|x_{t-1}) \exp [-\beta(u_{t-1},x_{t-1})]}{\gamma(x_{t-1})},
	\end{equation}
	where,
	\begin{equation}\label{gamma2}
		\gamma(x_{t-1}) = \int {^Ic(u_{t-1}|x_{t-1}) \exp[{-\beta(u_{t-1},x_{t-1})}] \mathrm{d} u_{t-1}},
	\end{equation}
	\begin{equation}\label{beta}
		\beta(u_{t-1},x_{t-1}) = \int s(x_{t}|u_{t-1}, x_{t-1}) s(o_{t}|x_{t})\times \ln \bigg ( \frac{s(o_{t}|x_{t}) }{ ^I s(o_{t}|x_{t}) } \frac{1}{{\gamma}(x_{t})}\bigg) \mathrm{d} x_{t}\mathrm{d} o_{t}.
	\end{equation}
\end{theorem}
\begin{proof}
 This theorem can be easily proven by adapting the proof of
Proposition 2 in \cite{Karny_2}.
\end{proof}
\begin{remark}
The solution of the control problem stated in theorem (\ref{Prop1}) is not restricted by the distributions characterising the evolution of the quantum system dynamics, their controllers, or their corresponding ideal distributions. It provide a general solution for any arbitrary pdfs. However, as will be seen in the following section, if all of the generative probabilistic models of the system dynamics, controller and ideal outcomes are Gaussian pdfs,  an analytic form for the randomised controller can be obtained.  
\end{remark}
\section{Solution of Quantum Control Problems with Gaussian pdfs }\label{QuantumCA}
In this section, we apply the theory introduced in Section (\ref{SFPDCQS}) for general quantum systems to quantum systems that can be described by Gaussian pdfs. By assuming that  the noise $\zeta_t$, appearing in equation (\ref{Inpu2NoiseDelay}), that is accounting for the different sources of uncertainties is Gaussian noise, the pdf of the system state, $x_t$ can be characterised by a complex normal pdf, 
\begin{align}
	\label{eq:eq4}
	s(x_{t}\left| {x_{t - 1}},u_{t-1} \right.) &\sim \mathcal{N}_{\mathcal{C}}(\mu_{t},{\Gamma}^{'},C^{'}), 
\end{align}
where,
\begin{align}
	\label{eq:relation1}
	\mu_t& =E(x_t)=Ax_{t-1}+Bu_{t-1},\nonumber\\ 
	\Gamma^{'}&= E((x_t-\mu_t)(x_t-\mu_t)^\dagger),\nonumber\\ 
	C^{'}&= E((x_t-\mu_t)(x_t-\mu_t)^T),
\end{align}
and where $A$ and $B$ are the state and control matrices respectively as stated in the previous sections, $E(.)$ stands for the expected value, $T$ is the  transpose operator and $\dagger$ is the conjugate transpose operator. The matrices $\mu_t$, $\Gamma^{'}$ and $C^{'}$ are called respectively the mean, the covariance and the relation or pseudo-covariance matrices. For completeness, the form of the complex normal distribution is recalled in Appendix (\ref{AppA_0}).

 Similarly, by assuming that the noise term $\sigma_t$ appearing in equation (\ref{outp}) is generated by a Gaussian distribution, the pdf of the measurement $o_t$ can be characterised by a Gaussian pdf as follows,
\begin{align}
	\label{eq:eq5}
	s({o_t}\left| {{x_{t}}} \right.) &\sim \mathcal{N}(Dx_{t},G),
\end{align}
with mean equal to $Dx_t$ and covariance matrix $G$ defined as follows,
 \begin{equation}
G=E((o_t-Dx_t)(o_t-Dx_t)^T).
\end{equation}
It is worth stating that the variable $o_t$ is real and for this reason, its pdf  is assumed to be a real Gaussian distribution. The whole system, the quantum system state, measurement and the electric field, i.e., the controller $u_t$, can be totally described by the following joint pdf,
\begin{equation}\label{gpdf2}
	f(x_t,o_t,u_{t-1}|x_{t-1})=c(u_{t-1}|x_{t-1})\mathcal{N}(Dx_{t},G)\mathcal{N}_{\mathcal{C}}(\mu_{t},{\Gamma}^{'},C^{'}),
\end{equation}
where $c(u_{t-1}|x_{t-1})$ is the distribution of the electric field as stated in the previous sections.\\
The desired distributions of the system state, measurement, and controller are then assumed to be given by the following ideal joint pdf, 
\begin{equation}\label{igpdf2}
	^If(x_t,o_t,u_{t-1}|x_{t-1})=	{}^Is(x_{t} | x_{t-1},{u_{t-1}}){}^Is({o_t}\left| {{x_{t}}} \right.)	{}^Ic({u_{t-1}}\left| {x_{t - 1}} \right),
\end{equation}
where,
\begin{align}
\label{eq:eq1st}
		{}^Is(x_{t} | x_{t-1},{u_{t-1}})& =s(x_{t} | x_{t-1},{u_{t-1}})\sim  \mathcal{N}_{\mathcal{C}}(\mu_{t},{\Gamma}^{'},C^{'}),\\
	\label{eq:eq2}
	{}^Is({o_t}\left| {{x_{t}}} \right.) &\sim \mathcal{N}({o}_{d},G_r),\\
	\label{eq:eq3}
	{}^Ic({u_{t-1}}\left| {x_{t - 1}} \right.) &\sim \mathcal{N}(u_r,\Omega),
\end{align}
and where $o_d$ and $G_r$ are respectively the mean and the covariance of the ideal pdf of the measurement, and $u_r$ and $\Omega$  stand for the mean and the covariance of the ideal pdf of the controller, respectively. Also, as can be seen from equation (\ref{eq:eq1st}), the ideal distribution of the quantum system state is taken to be equal to the actual distribution of the state as can be characterised from the Liouville-von Neumann equation. 
 
The control objective can then be stated as follows: Design a randomised controller, $c(u_{t-1}|x_{t-1})$ that will make the joint distribution of the quantum system as specified by equation ($\ref{gpdf2}$) as close as possible to the desired joint distribution as specified by equation (\ref{igpdf2}). The form of this distribution will be shortly shown, but we first provide the form of the performance index leading to the following theorem.

\begin{theorem}\label{Theo1}
	By substituting the ideal distribution of the measurements (\ref{eq:eq2}), the ideal distribution of the controller (\ref{eq:eq3}), the real distribution of  measurement (\ref{eq:eq5}), and the real distribution of  system dynamics (\ref{eq:eq4}) into (\ref{gamma2}), the performance index can be shown to be given by,
	\begin{align}
		- \ln \left( {\gamma \left( {{x_{t-1}}} \right)} \right) &= 0.5{x}_{t-1}^{T}{M_{t-1}}{x_{t-1}}+0.5P_{t-1}{{x}_{t-1}} + {0.5\omega_{t-1}}, \label{49}
	\end{align}
	where, 
	\begin{align}
		\label{Mt}
		{M_{t - 1}} &=A^{T} (D^{T} G_r^{-1}D+M_{t})A-A^{T} (D^{T} G_r^{-1}D+M_{t})^{T} B
		\big(\Omega^{-1}+B^{T} ({D^{T} }G_r^{ - 1}D+M_{t})B\big)^{-1}
		\nonumber\\
		&	B^{T} (D^{T} G_r^{-1}D+M_{t})A,\end{align}

	\begin{align}
		\label{Pt}
		P_{t-1}&=(P_{t}-2{o}_{d}^{T} G_r^{-1}D)A+2(\Omega^{-1}u_r-0.5B^{T} (P_{t}^{T} -2D^{T} G_r^{-1}{o}_{d}))^{T} 
		\big(\Omega^{-1}+B^{T} ({D^{T} }G_r^{ - 1}D+M_{t})B\big)^{-1}\nonumber\\
		&B^{T} (D^{T} G_r^{-1}D+M_{t})A,
	\end{align}
	and 
	\begin{align}
		\label{Const}
		\omega_{t-1}&={\omega_{t}}+ {o}_{d}^{T} G_r^{ - 1}{{o}_{d}}   +\ln(\dfrac{|G_r|}{|G|})- \Tr\big(G(G^{-1}-G_r^{ - 1})\big)-\Tr(\bar{C}^{-1}({D^{T} }G_r^{ - 1}D+M_{t}))
		+u_r^{T} \Omega^{-1}u_r	\nonumber\\
		&-\big( \Omega^{-1}u_r-0.5B^{T} (P_{t}^{T} -2{D}^{T}G_r^{ - 1}{{o}_{d}})\big)^{T} 
		\big(\Omega^{-1}+B^{T} ({D^{T} }G_r^{ - 1}D+M_{t})B\big)^{-1}	
		\big( \Omega^{-1}u_r-0.5B^{T} (P_{t}^{T} -2{D}^{T}G_r^{ - 1}{{o}_{d}})\big)\nonumber\\
		&-2\ln({|\Omega|^{-1/2}}|\Omega^{-1}+B^{T} ({D^{T} }G_r^{ - 1}D+M_{t})B|^{-1/2}),
	\end{align}	
such that $C=\dfrac{C^{'}}{|\Gamma^{'}|^2-|C^{'}|^2}$, and $|G|$ stands for the determinant of the matrix $G$. 
\end{theorem}
\begin{proof}
	The proof of this theorem is given in appendix (\ref{AppA}).
\end{proof}
\begin{remark}
For the performance index (\ref{49}) to be real, the $M_{t - 1}$ matrix needs to be evaluated as the product of the transpose of a vectorised Hermitian operator with the vectorised Hermitian operator itself. To elaborate, consider a vectorised Hermitian operator, $W_{t-1}$, then $M_{t - 1}=W^T_{t-1}W_{t-1}$ . This means that the first term in equation (\ref{49}) will be real. Similarly, $P_{t-1}$ in equation (\ref{49}) is a vectorised Hermitian operator.
\end{remark}
Having specified the form of the performance index (\ref{49}), we can easily find the distribution of the optimal controller, $c(u_{t-1}|x_{t-1})$ from definition (\ref{eq:Eqn4}). This is given in the following theorem.
\begin{theorem}\label{Theo2}
	The distribution of the controller  that  minimises the recurrence equation (\ref{cost_to_go_2}) is given by, 
	\begin{align}\label{Probcapp}
c({u_{t-1}}\left| {x_{t - 1}} \right.) &\sim \mathcal{N}(v_{t-1},R_t),
	\end{align}
where, 
\begin{align}\label{optimalcapp}
	v_{t-1}&=\bigg({\Omega ^{ - 1}} + {B^T}( {D^T}G_r^{ - 1}D + {M_t})B\bigg)^{-1} \bigg({\Omega ^{ - 1}}{u_r} - {B^T}({D^T}G_r^{ - 1}D + {M_t})A{x_{t - 1}} -0.5 {B^T}( P_{t }^T - 2{D^T}G_r^{ - 1}{{o}_{d}} )\bigg),
\end{align}
and,
\begin{align}\label{optimal_variance}
	R_t&=\bigg({\Omega ^{ - 1}} + {B^{T} }( {D^{T} }G_r^{ - 1}D + {M_t})B\bigg)^{-1}.
\end{align}
\end{theorem}
\begin{proof}
The proof is given in Appendix (\ref{appC}).
\end{proof}
\begin{remark}
The derived probabilistic controller (\ref{Probcapp}) for the assumed Gaussian probability distributions maintains the standard form of linear quadratic controllers. However,  the probabilistic controller (\ref{Probcapp}) is more exploratory due to its probabilistic nature. Ideally control inputs should be sampled from the obtained pdf of the randomised controller. This however results in slightly worse control quality, but randomisation makes the controller more explorative.
\end{remark}
\subsection{Implementation of the Gaussian controller}
The step by step implementation procedure of the Gaussian controller provided in Theorem (\ref{Theo2}) is summarised as pseudo-code in the following algorithm.

\begin{algorithm}[H]
	\caption{Fully probabilistic control of quantum systems\hspace{4cm}}\label{alg:cap} 
	\begin{algorithmic}[1]
		\State  $\text{Evaluate the operator} \hspace{0.1cm} D\hspace{0.1cm} \text{associated with the target operator}\hspace{0.1cm} \hat{o}\text{;}$
		\State  $\text{Compute the matrices} \hspace{0.1cm} \tilde{A}\hspace{0.1cm} \text{and} \hspace{0.1cm} \tilde{N}\hspace{0.1cm} \text{from equation (\ref{rho_elements}), hence evaluate}\hspace{0.1cm} A\hspace{0.1cm} \hspace{0.1cm} \text{from equation (\ref{A_B_t})} \text{;}$
			\State  $\text{Determine the predefined  desired value ${o}_{d}$;}$
		\State  $\text{Specify the initial state $x_0$, the shifting state, $x_e$, and then calculate the initial value of the measurement}$ $\text{ state $o_0\gets Dx_0$;}$
		\State  $\text{Provide the covariance of the ideal distribution of the state vector, $G_r$ and the covariance of the controller, $\Omega$;}$
		\State  $\text{Initialise:}$ $t\gets0$, $M_0 \gets \text{rand}$, $P_0 \gets \text{rand}\text{;}$
		\While{$t \ne \mathcal{H}$}
		\State  $\text{Evaluate} \hspace{0.1cm} B\hspace{0.1cm} \hspace{0.1cm} \text{from equation (\ref{A_B_t})} \text{;}$
		\State $\text{Calculate the steady state solutions of} \hspace{0.1cm} M_t, \text{and}\hspace{0.1cm}  P_t\hspace{0.1cm}\text{following the formulas provided in  equations (\ref{Mt}),}$ $\text{ and (\ref{Pt}) respectively};$
		\State $\text{Use}\hspace{0.05cm} M_t, \text{and}\hspace{0.05cm} P_t\hspace{0.05cm} \text{ to compute the optimal control input, $v_{t-1}$  following  equation (\ref{optimalcapp}) given in Theorem (\ref{Theo2});} $ 
		\State  $\text{Set:}$ $u_{t-1}\gets v_{t-1}$; 
		\State $\text{Using the obtained control input from the previous step, evaluate \hspace{0.1cm}  $x_t $ according to equation (\ref{Inpu2NoiseDelay}),}$ $\text{${x}_{t}\gets Ax_{t-1}+B(x_{t-1}) u_{t-1} + \zeta_{t} $} $
		\State $\text{Following  equation (\ref{outp}), evaluate $o_t$ to find the measurement state at time instant $t$, $o_t \gets Dx_t+\sigma_t$ ;}$
		\State $t\gets t+1\text{;}$	
		\EndWhile
	\end{algorithmic}
\end{algorithm}
Algorithm (\ref{alg:cap}) will be applied in the next section to control a set of quantum physical systems.
\section{Results and discussions}\label{Results and discussions}
The  fully probabilistic control approach introduced in the previous sections will be applied in this section to control certain quantum physical systems and transform their outcome according to prespecified target operators. In particular, we will study a molecular system whose vibrations are described by the Morse potential and spin systems interacting with an external electric field.
\subsection{Molecular system described by the Morse potential}
In this section a  molecular system whose vibrations are described by the Morse potential will be used to demonstrate the proposed Fully probabilistic quantum control method. We start by providing the general theory of Morse potential and then present the simulation result of controlling a 3-level Lithium hybrid molecule. 
\subsubsection{Morse Potential}
The Morse potential model was introduced by P. Morse in 1929 in order to study the motions of the nuclei in  diatomic molecules\cite{Morse}. Since then, the Morse potential has been exploited in many different areas of physics, chemistry,  and biology \cite{Zdravkovic}. Particularly, it was used to describe the vibrations inside diatomic molecules  and in the study of  interactions of molecules with electromagnetic fields\cite{Pauling,Herzberg}. 

The one-dimensional Morse potential describing the vibrations of two oscillating atoms is given by,
\begin{equation}\label{MorseHamilt}
	V_M(r)=D_{0}(e^{-2\alpha (r-r_{\text{eq}})}-2e^{-\alpha (r-r_{\text{eq}})}),
\end{equation}
where  $ r $  is the displacement of the two atoms from their equilibrium positions $ r_{\text{eq}}$. $D_{0}$
represents the depth of the potential well at the equilibrium $r_{\text{eq}}$, while  $\alpha$ is related to the width of the potential. The  energy eigenvalues of the Morse potential are given as follows, 
\begin{equation}\label{Morsespec}
	E_{n}=-\dfrac{\hbar^2\alpha^2}{2m_r}(\dfrac{\nu-1}{2}-n)^{2},
\end{equation}
where $\nu$ is a parameter related to the diatomic molecule, and $m_r$ is the reduced mass of the oscillating atoms. The particularity of the Morse potential compared to the harmonic potential is that it has a finite number of bound states, i.e., $ n=0,1,\dots ,l-1$, where $l-1=[\frac{\nu-1}{2}]$, and where $[\nu]$ is the integer part of $\nu$. Here, the free Hamiltonian $H_0$  appearing in equation (\ref{LVN_eq}) has the following form,
\begin{equation}
H_0=\sum_{n=0}^{l-1}E_n\ket{n}\bra{n},
\end{equation}
where $\ket{n}$ is the energy eigenstate associated with the energy eigenvalue $E_n$. The interaction between the molecule and the electric field is modelled  as follows,
\begin{equation}
H_u(t)=-\mu u_t,
\end{equation}  where $\mu$ is the electric dipole  given in the position representation $\{\ket{r}\}$  by  $\mu(r)=\mu_0 r\text{e}^{-r/r^{*}}$ with  $ \mu_0 $ and $r^*$ being parameters related to the molecule. The matrix elements of $\mu$  can be computed as follows,
\begin{equation}\label{electricdipol}
\mu_{n,m}=\bra{n}\mu\ket{m}=\int_{-\infty}^{\infty} {\mu}(r)\psi_{n}^{\nu}(r)  \psi_{m}^{\nu}(r){d}{r},
\end{equation}
where $\psi_{n}^{\nu}(r)$ is the  eigenfunction associated with the energy eigenvalue $E_n$ and is given by, 
\begin{equation}\label{psiMorse}
\braket{r|n}=	\psi_{n}^{\nu}(r)=\text{N}_{n}e^{-\frac{y}{2}}y^{s}L_{n}^{2s}(y),
\end{equation}
where we have used the change of variable  $y=\nu e^{-\alpha (r-r_{\text{eq}})}$, and where $L_{n}^{2s}(y)$ are the Laguerre polynomials, $2s=\nu-2n-1$, and
$\text{N}_{n} $ is the normalization factor governed by, 
 \begin{equation}\label{Norm}
	\text{N}_{n}=\sqrt{\frac{\alpha(\nu-2n-1)\Gamma(n+1)}{\Gamma(\nu-n)}},
\end{equation}
where $\Gamma$ is the gamma function\cite{Morse}.
\subsubsection{Lithium Hybrid Diatomic molecule, $^7\text{Li}\hspace*{0.05cm}^2\text{H}$ }
The proposed probabilistic quantum control method is applied here to control the outcome of an observable associated with the Lithium hybrid molecule, $^7\text{Li}\hspace*{0.05cm}^2\text{H}$  whose vibrations are described by the Morse oscillator. In this experiment, we omit the effect of the environment in equation (\ref{LVN_eq}) and consider that the  molecule only interacts with the electric field $u_t$. The Morse parameters associated with this molecule and that are used here can be found in\cite{Herzberg}. Their values are given by  $D_0=2.45090 \hspace{0.1cm}\text{eV}$, $r_{\text{eq}}=2.379 \hspace{0.1cm}\text{{\AA}} $, $m_r=2.5986\times10^{-27}\hspace{0.2cm}\text{Kg}$,
$\alpha=13.956\hspace{0.2cm}\text{{\AA}}^{-1}$, and $\nu\approx 6.1346$. Hence, according to these parameters, the quantum number $n$ takes three values $n=0,1,2$ and the density operator $\rho$ associated with the system is an $l\times l$-matrix, where $l=3$. The  parameters related to the electric dipole can also be found in \cite{Herzberg}. Their values are given by $\mu_0\approx5.8677\hspace{0.2cm} \text{Debye}$, and $r^*\approx1.595\text{{\AA}}$.

As discussed in earlier sections, the quantum system state is described by the vector $\tilde{x}_t$ which is the vectorisation of the density matrix recalled in equation (\ref{x_t_vet}). At the initial time, we assume that the system is prepared in the ground state. This implies that, \begin{equation}\tilde{x}_0=[1 \hspace{0.2cm}0\hspace{0.2cm} 0\hspace{0.2cm}  0\hspace{0.2cm} 0\hspace{0.2cm} 0\hspace{0.2cm} 0\hspace{0.2cm} 0\hspace{0.2cm} 0].
\end{equation}
By omitting the effect of the environment, the evolution of the density matrix elements provided in equation (\ref{rho_elements}) can be simplified as,
	\begin{equation}
	\label{rho_elements_NoEnv}
	\dfrac{d\rho_{n,m}(t)}{d t}=-i\omega_{n,m}\rho_{n,m}(t)+i\dfrac{u(t)}{\hbar}\sum_{k=0}^{l-1=2}(\mu_{n,k}\rho_{k,m}(t)-\rho_{n,k}(t)\mu_{k,m}),
\end{equation} 
where in this example, $\{n,m\}=\{0,1,2\}$, $\omega_{n,m}=\dfrac{E_n-E_m}{\hbar}$ such that $E_n$ are the Morse energy eigenvalues recalled in equation (\ref{Morsespec}), and the matrix elements of the electric dipole,  $\mu_{n,m}$ can be computed using equation  (\ref{electricdipol}). Using the definition of vectorisation of the density matrix provided in equation (\ref{x_t_vet}) along with equation (\ref{rho_elements_NoEnv}) and equation  (\ref{electricdipol}), the parameters of the state equation defined in equation (\ref{state_space}) can be easily obtained as follows, 
\begin{equation}
\tilde{A}=\text{diag}[0 \hspace{0.2cm}0\hspace{0.2cm}0\hspace{0.2cm} -i\omega_{0,1}\hspace{0.2cm}  -i\omega_{0,2}\hspace{0.2cm} -i\omega_{1,0}\hspace{0.2cm} -i\omega_{2,0}\hspace{0.2cm} -i\omega_{1,2}\hspace{0.2cm}-i\omega_{2,1}]
\end{equation}
 \begin{equation}
\tilde{N}=\dfrac{1}{\hbar}\left(
\begin{array}{ccccccccc}
	0 & 0 & 0 &-\mu_{0,1} & -\mu_{2,0} & \mu_{0,1} &\mu_{0,2} & 0 & 0 \\
	0 & 0 & 0 &\mu_{1,0} & 0 & -\mu_{0,1} & 0 &- \mu_{2,1} & \mu_{1,2}\\
	0 & 0 & 0 & 0 & \mu_{2,0} & 0 & -\mu_{0,2} &\mu_{2,1} &- \mu_{1,2} \\
	-\mu_{0,1} & \mu_{0,1} & 0 & \mu_{0,0}-\mu_{1,1} & -\mu_{2,1} & 0 & 0 & 0 & \mu_{0,2} \\
	-\mu_{0,2} & 0 &\mu_{0,2} & -\mu_{1,2} & \mu_{0,0}-\mu_{2,2}  & 0 & 0 & \mu_{0,1}  & 0 \\
	\mu_{1,0}  & -\mu_{1,0} & 0 & 0 & 0 &-\mu_{0,0}+\mu_{1,1} & \mu_{1,2} & -\mu_{2,0} & 0 \\
\mu_{2,0} & 0 & -\mu_{2,0} & 0 & 0 & \mu_{2,1} & -\mu_{0,0}+\mu_{2,2} & 0 & -\mu_{1,0} \\
	0 & -\mu_{1,2}  & \mu_{1,2}  & 0 & \mu_{1,0}  & -\mu_{0,2}  & 0 & \mu_{1,1} -\mu_{2,2} & 0 \\
	0 & \mu_{2,1}  & -\mu_{2,1} & \mu_{2,0} & 0 & 0 & -\mu_{0,1} & 0 & -\mu_{1,1}+\mu_{2,2}\\
\end{array}
\right),\end{equation}
where $\tilde{A}$ is a diagonal $9\times9$-matrix, and the shifting vector $x_e=\tilde{x}_0$. These parameters are then used in equation (\ref{A_B_t}), with $\Delta{t}=0.0167$ to obtain the system matrices, $A$, and $B$ that are required for the evaluation of the optimal control signal. Furthermore, the matrix $D$ given in equation (\ref{outp}), which is also required for the evaluation of the optimal control signal can be obtained by evaluating the target operator which is taken in this example to be a Gaussian operator \cite{RBZ98_2}, defined as follows,
\begin{equation}\label{targetoper}
	\hat{o}=o(r)=\dfrac{\gamma_0}{\sqrt{\pi}}\text{e}^{-\gamma_0^2(r-r^{'})^2},
\end{equation}
with $\gamma_0=47.2590 \hspace{0.2cm}\text{{\AA}}^{-1}
$, and $r^{'}=2.4871 \hspace{0.1cm}\text{{\AA}}$. The matrix representation of the target operator (\ref{targetoper}) can then be obtained as follows,
\begin{equation}
o_{i,j}=\int_{-\infty}^{\infty} {o}(r)\psi_{i}^{\nu}(r)  \psi_{j}^{\nu}(y){d}{r},\hspace{0.3cm}\text{where}\hspace{0.3cm}\{i,j\}=\{0,1,2\}.
\end{equation} 
Therefore, following the vectorisation method described in equation (\ref{x_t_vet}), the operator $D$ that appears in the measurement equation (\ref{outp}) will have the following form,
\begin{equation}
D=[o_{0,0}\hspace{0.2cm} o_{1,1}\hspace{0.2cm}o_{2,2}\hspace{0.2cm}o_{0,1}\hspace{0.2cm}o_{0,2}\hspace{0.2cm}o_{1,0}\hspace{0.2cm}o_{2,0}\hspace{0.2cm}o_{1,2}\hspace{0.2cm}o_{2,1}].
\end{equation}
\begin{figure}[t]
	\begin{subfigure}{.5\textwidth}
		\centering
		\includegraphics[width=.8\linewidth]{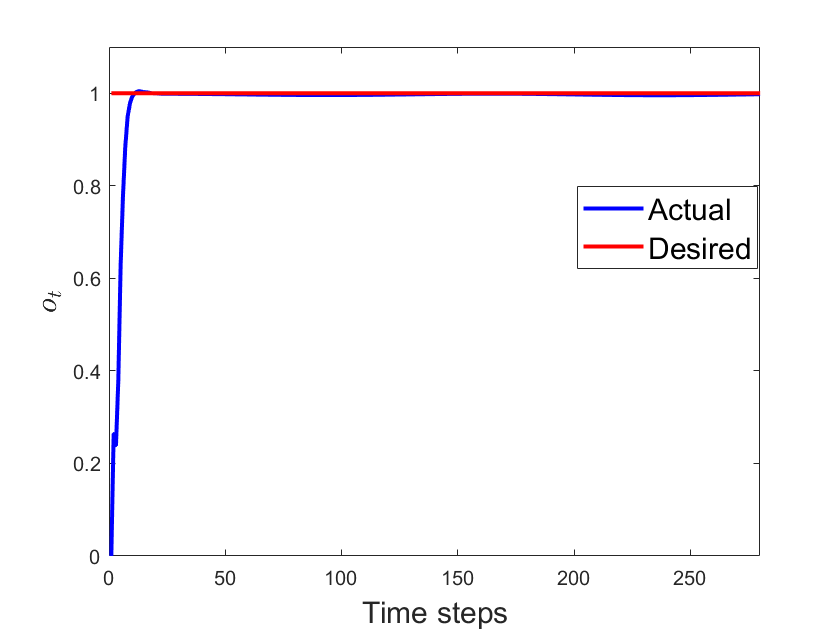}  
		\caption{}
		\label{fig:sub-first}
	\end{subfigure}
	\begin{subfigure}{.5\textwidth}
		\centering
		\includegraphics[width=.8\linewidth]{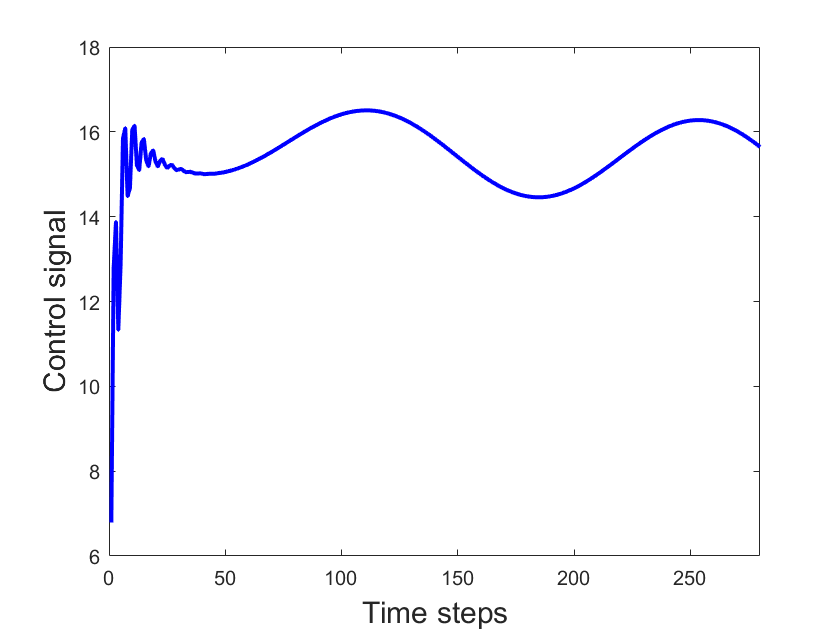}  
		\caption{}
		\label{fig:sub-second}
	\end{subfigure}
	\caption{(\ref{fig:sub-first} ): The blue curve represents time evolution of the measurement output, $o_t$ of the considered Morse potential associated with the Lithium hydride molecule $^7\text{Li}\hspace*{0.05cm}^2\text{H}$  under the effect of the control signal, $u_t$.  The red curve is the desired value of the Gaussian operator, i.e., ${o}_{d}=1$.  (\ref{fig:sub-second}): The time evolution of  the control signal,  $u_t$ responsible for transforming the system measured output, $o_t$ to the predefined desired output, ${o}_{d}=1$. In this figure atomic units were used.}
	\label{figure_1}
\end{figure}
Following the computation of  the $A$, $B$ and $D$ operators, and taking $G_r=0.00000001$, and $\Omega=0.28950$, the matrices $M_t$ and $P_t$ defined in equations (\ref{Mt}) and (\ref{Pt}), respectively are evaluated at each instant of time as discussed in Algorithm (\ref{alg:cap}). The $M_t$ and $P_t$ matrices are then used to compute the electric field, i.e., the control signal, $u_{t-1}$ using (\ref{optimalcapp}). These steps are repeated until the measurement output, $o_t$ became as close as possible to the  predefined desired value ${o}_{d}$, which is taken to be equal to $1$ in this experiment, i.e., ${o}_{d}=1$.

The result of this experiment is shown in Figure (\ref{fig:sub-first}). As can be clearly seen from this figure, the output, $o_t$ of the considered Morse potential associated with the Lithium hydride molecule $^7\text{Li}\hspace*{0.05cm}^2\text{H}$ increases quickly and reaches in a few time steps the predefined desired value, ${o}_{d}=1$ under the effect of the derived electric field  and then stabilises at that value. This shows that the specified control objective is achieved and demonstrates the effectiveness of our proposed method. Here, as can be seen from the figure, the control objective is achieved and kept for longer times steps. Figure (\ref{fig:sub-second}) shows the electric field that is obtained according to Algorithm (\ref{alg:cap}) and that is responsible for achieving the control objective. 
\subsection{Spin systems}
This section demonstrates the effectiveness of the proposed control method in controlling quantum systems with 
spin $j$, and, more generally, any quantum system with three
observables $J_i$ $(i = 1,2,3)$ satisfying the angular momentum
commutation relations $[J_1,J_2]=i\hbar\epsilon_{123}J_3$ (with $\epsilon_{123}$ being the Levi-civita symbol). 

A spin-j state can be described by,
\begin{equation}
\ket{\psi}=\sum_{m=-j}^{j}c_m\ket{j,m},
\end{equation}
where $c_m$ are complex coefficients that satisfy $\sum_{m=-j}^{j}|c_m|^2=1$, and the basis $\{\ket{j,m}\equiv\ket{m}, m=-j,\dots,j\}$ is formed by the common eigenstates of the operators $J^2=J_1^2+J_2^2+J_3^2$ and $J_3$.  Without loss of generality, we consider that the spin system interacts only with an external electric field $u_t$. Additionally, the Hamiltonian of the system is taken to be given by, 
\begin{equation}
H=\hbar\omega H_0+\theta H_u(t),
\end{equation}
where $\omega$ and $\theta$ are parameters related to the spin system, and $H_0$ and $H_u(t)$ are some Hermitian operators \cite{spin}. In the following, we take $\hbar\omega=\theta=1$ for simplicity reasons.
 The objective now is to transfer the system initially prepared in a given state $\ket{\psi}_q$ to another state $\ket{\psi}_g$ using the interaction between the system and the external electric field, thus deriving an optimal control sequence that allows that transition. For a projector operator, the output $o_t$ at each time instant $t$ is nothing but the squared overlap between the state $\ket{\psi(t)}$ and the  state $\ket{\psi_g}$.  Thus, the control objective of transferring the system from the state $\ket{\psi}_q$ to the state $\ket{\psi}_g$ is equivalent to maximizing the squared overlap between the states $\ket{\psi(t)}$ and $\ket{\psi_g}$. This means that we should take $o_d=1$.
 
Thereafter, we consider spin-1/2 system, i.e., $j=\dfrac{1}{2}$ and spin-1 system, i.e., $j=1$.

\subsubsection{spin-1/2 system }
The proposed control method is applied in this section to a spin-1/2 system interacting with an external electric field, $u_t$ with the interaction being modelled as follows,
\begin{equation}\label{58_}
H=H_0+H_u(t)=\dfrac{1}{2}\sigma_3+\dfrac{1}{2}(\sigma_1+\sigma_2)u_t,
\end{equation} 
where  $\sigma_1$, $\sigma_2$ and $\sigma_3$ are the Pauli matrices given in the basis $\{\ket{-}\equiv\ket{\dfrac{1}{2},-\dfrac{1}{2}},\ket{+}\equiv\ket{\dfrac{1}{2},\dfrac{1}{2}}\}$ as,
\begin{equation}\label{PauliMat}
\sigma_1	=\left(
\begin{array}{cc}
	0&1\\
	1 & 0 \\
\end{array}
\right),\hspace{0.4cm}	\sigma_2=\left(
\begin{array}{cc}
0 &-i\\
	i & 0 \\
\end{array}
\right),\hspace{0.4cm}\sigma_3	=\left(
\begin{array}{cc}
	1 &0\\
	0 & -1\\
\end{array}
\right).
\end{equation}
The control objective in this example is to find an optimal control sequence that can transform the system from the initial state $\ket{-}$ to the state $\ket{+}$ using the proposed probabilistic quantum control approach. Hence, the target operator is taken to be $D=[0\hspace{0.2cm}1 \hspace{0.2cm}0\hspace{0.2cm}0]$, which is nothing but $D=(\text{vect}(\Pi_+))^T$, where $\Pi_+=\ket{+}\bra{+}$.

Furthermore, using the operators, $\tilde{A}$ and $\tilde{N}$ whose calculations along with their values are given in details in Appendix (\ref{spinn_1_2}), the  $A$, and $B$ matrices that are required for the evaluation of the optimal control signal are evaluated using equation (\ref{A_B_t}), with $\Delta{t}=0.0505$. Following the computation of  the $A$, $B$ and $D$ operators, and taking $G_r=0.00001$, and $\Omega=1.0$, the matrices  $M_t$ and $P_t$ defined in equations (\ref{Mt}) and (\ref{Pt}), respectively are evaluated at each instant of time as discussed in Algorithm (\ref{alg:cap}). The $M_t$ and $P_t$ matrices are then used to compute the electric field, i.e., the control signal, $u_{t-1}$. These steps are repeated until the measurement output, $o_t$ became as close as possible to the  predefined desired value ${o}_{d}$, which is evaluated to be equal to $1$ in this experiment.

Figure (\ref{spin3}) shows the time evolution of the population of the $\ket{+}$ state, $\rho_{11}(t)$ for a spin-1/2 system initially prepared in the state $\ket{-}$ as a result of its interaction with the derived  electric field as calculated from equation (\ref{optimalcapp}). Figure (\ref{spin1_4}) shows the time evolution of the optimal control signal. From Figure (\ref{spin3}), it can be clearly seen that the population $\rho_{11}$ have reached in a few time steps the predefined desired target value ${o}_{d}=1$. 
\begin{figure}[h]
	\begin{subfigure}{.5\textwidth}
		\centering
		\includegraphics[width=.8\linewidth]{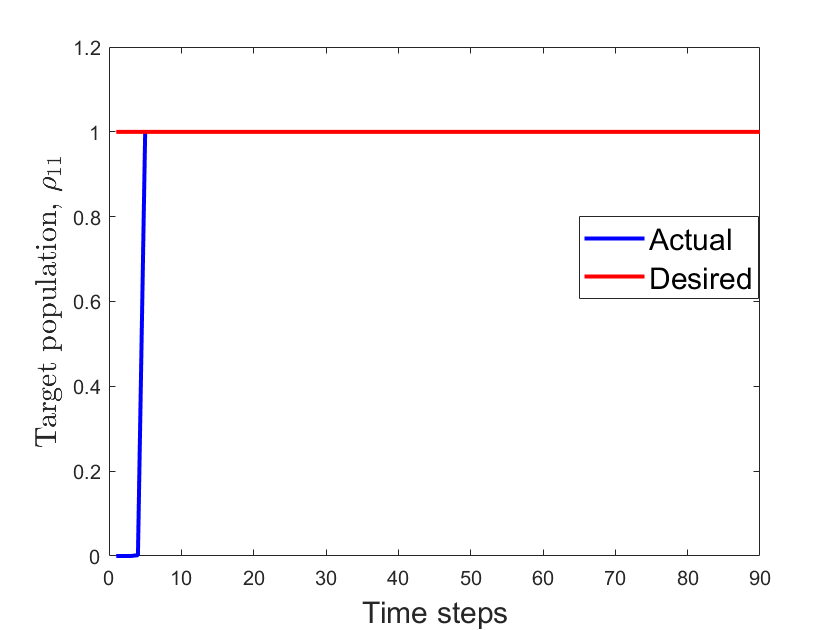}  
		\caption{}
		\label{spin3}
	\end{subfigure}
	\begin{subfigure}{.5\textwidth}
		\centering
		\includegraphics[width=.8\linewidth]{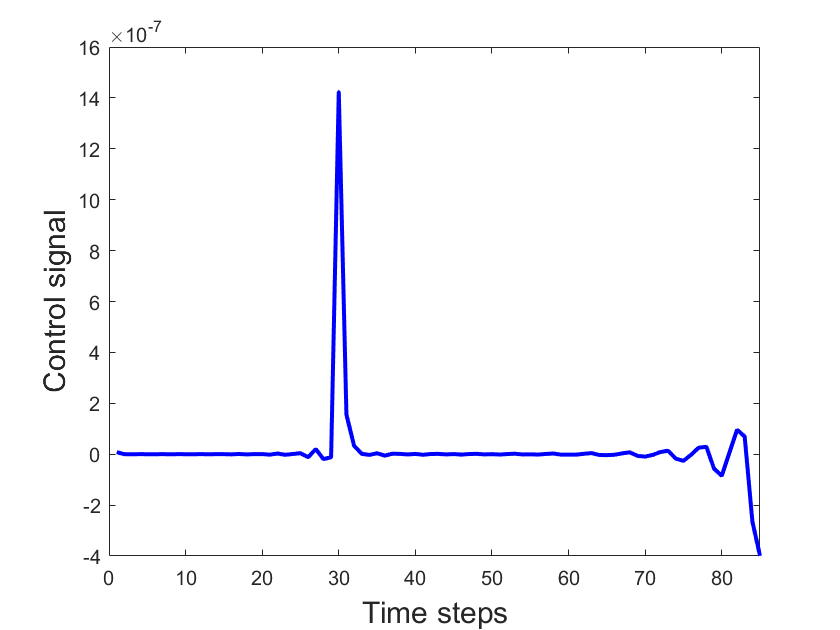}  
		\caption{}
		\label{spin1_4}
	\end{subfigure}
	\caption{(\ref{spin3} ): The blue curve represents the time evolution of the population of the $\ket{+}$ state, $\rho_{11}(t)$ of the considered spin-1/2 system. The red curve is the desired value ${o}_{d}=1$. (\ref{spin1_4}): Time evolution of  the control signal,  $u_t$ responsible of achieving the control objective.}
	\label{figure_2_2}
\end{figure}
\subsubsection{spin-1 system}
To further demonstrate the efficacy of the proposed probabilistic quantum control method,  it is applied here to control a spin-1 system initially prepared in the state $\ket{-1}$. For this spin-1 system, the Hamiltonian  describing the interaction of the system is taken to be given by, 
\begin{equation}
	H=H_0+H_u(t)=\left(\begin{array}{ccc}
		\dfrac{3}{2} & 0 & 0  \\
		0 & 1 & 0  \\
		0 & 0&0   \\
	\end{array}\right)+\left(\begin{array}{ccc}
		0 & 0 & 1  \\
		0 & 0 & 1  \\
		1 & 1&0   \\
	\end{array}\right) u(t).
\end{equation} 
Two sets of experiments were conducted. In the first experiment, the control objective is specified as to transfer the system to state $\ket{0}$. The target operator that corresponds to this objective is the projector $\Pi_0=\ket{0}\bra{0}$ which in vectorised form is given as $D=(\text{vect}(\Pi_0)^T)^T$ . Furthermore, using the operators, $\tilde{A}$ and $\tilde{N}$ whose calculations along with their values are given in details in Appendix (\ref{spinn_1}), the  $A$, and $B$ matrices that are required for the evaluation of the optimal control signal are evaluated using equation (\ref{A_B_t}), with $\Delta{t}=0.0067$. Following the computation of  the $A$, $B$ and $D$ operators, and taking $G_r=0.000000001$, and $\Omega=1$, the matrices  $M_t$ and $P_t$ defined in equations (\ref{Mt}) and (\ref{Pt}), respectively are evaluated at each instant of time as discussed in Algorithm (\ref{alg:cap}). The $M_t$ and $P_t$ matrices are then used to compute the electric field, $u_{t-1}$. These steps are repeated until the measurement output, $o_t$ became as close as possible to the  predefined desired value ${o}_{d}$, which is evaluated to be equal to $1$ in this experiment. Figure (\ref{figure5}) shows the time evolution of the population of the $\ket{0}$ state, $\rho_{11}(t)$ for the considered spin-1 system initially prepared in the state $\ket{-1}$ as a result of its interaction with the derived electric field as calculated from equation (\ref{optimalcapp}). The time evolution of the obtained optimal control signal that allowed the achievement of this control objective is given in Figure (\ref{figure6}). From these two figures, it can be clearly seen that the designed probabilistic controller was able to achieve the control objective and transfer the spin-1 system state from the initial state, $\ket{-1}$, to the desired state, $\ket{0}$.

In the second experiment, the control objective is specified as to transfer the system from the state $\ket{-1}$ to the state $\ket{1}$. In this experiment, the  $A$, and $B$ matrices are calculated  by substituting  the $\tilde{A}$ and $\tilde{N}$ as given in Appendix (\ref{spinn_1}) in equation  (\ref{A_B_t}) and using $\Delta{t}=0.0404$. By taking $G_r=0.000000001$, and $\Omega=0.11$, and following the procedure detailed in Algorithm (\ref{alg:cap}) the optimal control signal is calculated at each time instant  and used to control the spin-1 system. The time evolution of the population of the $\ket{1}$ state, $\rho_{22}(t)$ as a result of the obtained control signal is shown in Figure (\ref{figure7}), while the evolution of the optimal control is shown in Figure (\ref{figure8}). As can be seen from these figures the specified control objective is achieved and the obtained control input was able to transfer the spin-1 system state from state $\ket{-1}$ to state $\ket{1}$ and maintain it at that value. 

This demonstrates the effectiveness of the proposed method.
\begin{figure}[h!]
	\begin{subfigure}{.5\textwidth}
		\centering
		\includegraphics[width=.8\linewidth]{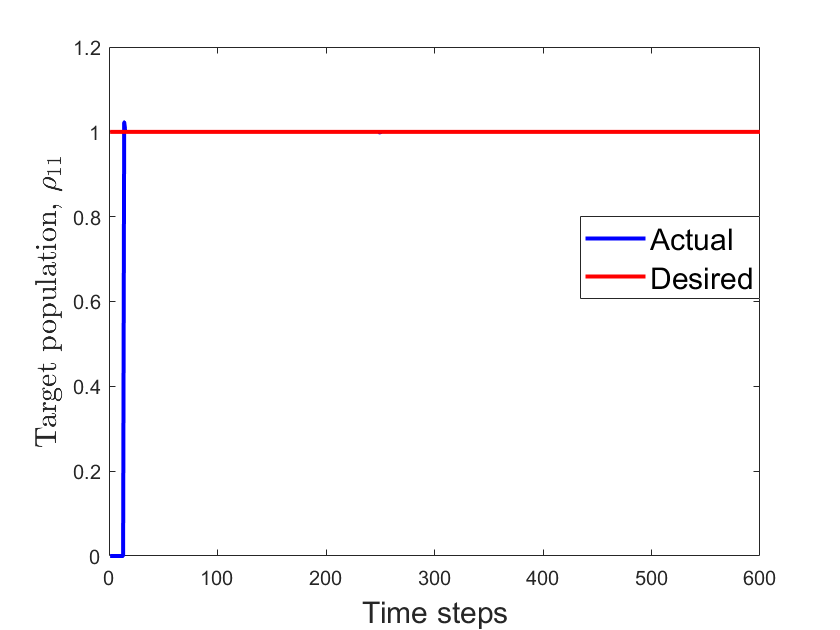}  
		\caption{}
		\label{figure5}
	\end{subfigure}
	\begin{subfigure}{.5\textwidth}
		\centering
		\includegraphics[width=.8\linewidth]{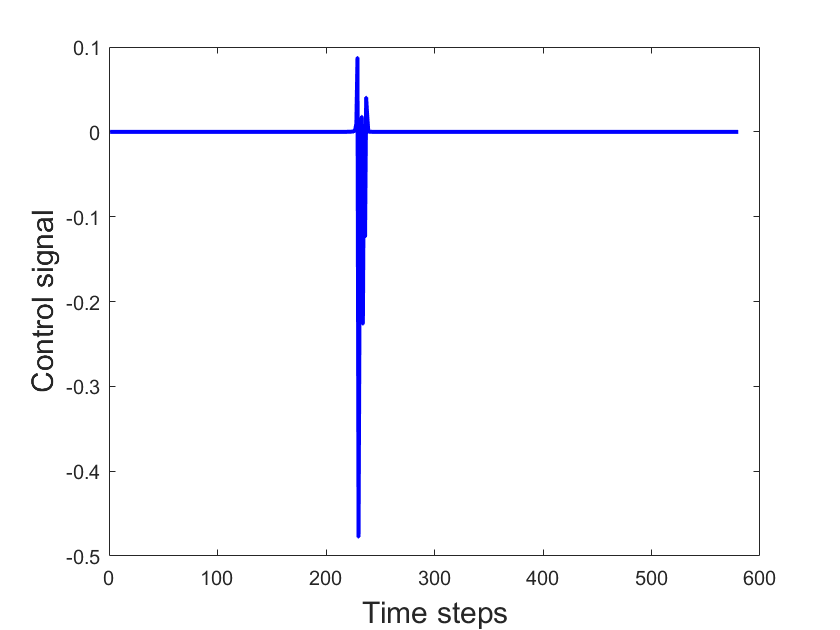}  
		\caption{}
		\label{figure6}
	\end{subfigure}\\
	\begin{subfigure}{.5\textwidth}
	\centering
	\includegraphics[width=.8\linewidth]{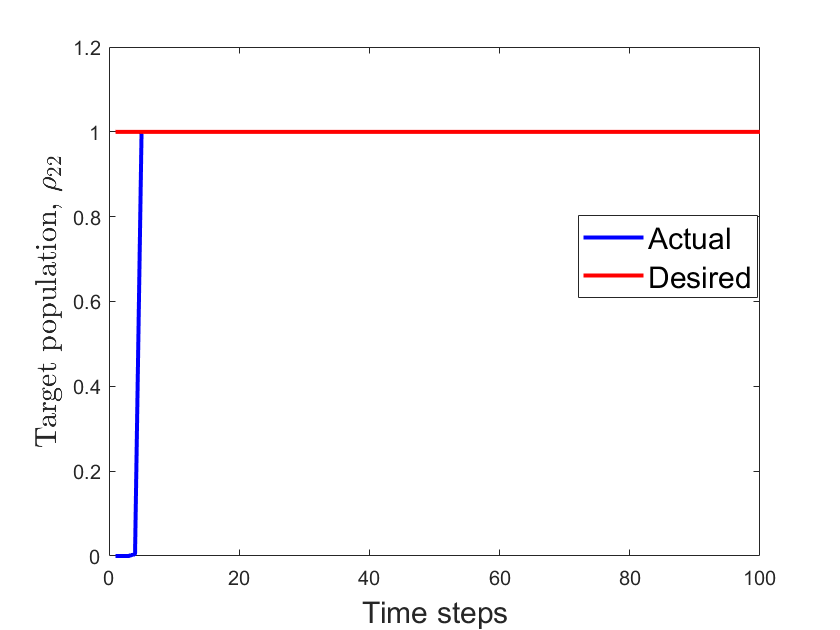}  
	\caption{}
	\label{figure7}
\end{subfigure}
\begin{subfigure}{.5\textwidth}
	\centering
	\includegraphics[width=.8\linewidth]{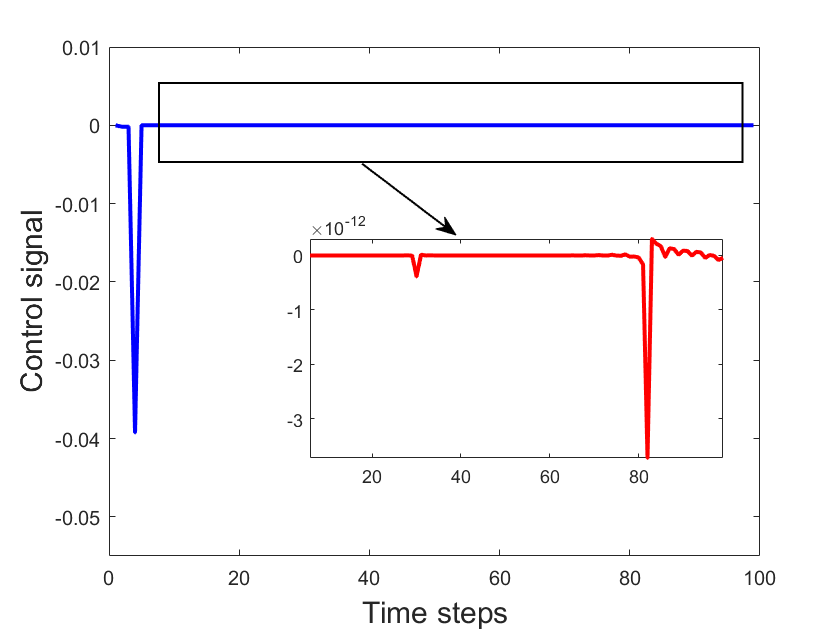}  
	\caption{}
	\label{figure8}
\end{subfigure}
	\caption{(\ref{figure5} ): The blue curve represents the time evolution of the population of the $\ket{0}$ state, $\rho_{11}(t)$ of the considered spin-1 system. The red curve is the desired value ${o}_{d}=1$. (\ref{figure6}): Time evolution of  the control signal,  $u_t$ responsible for achieving the control objective. (\ref{figure7} ):  The blue curve represents the time evolution of the population of the $\ket{1}$ state, $\rho_{22}(t)$ of the considered spin-1 system. The red curve is the desired value ${o}_{d}=1$.  (\ref{figure8}):  Time evolution of  the control signal,  $u_t$ responsible for achieving the control objective.}
	\label{figure_3}
\end{figure}
\section{Final comments}\label{conclu}
In this paper, we considered the control problem of quantum physical systems  using a probabilistic approach. First, we constructed the state space model associated with quantum systems and showed that their complete characterisation can be only obtained by pdfs. Considering the minimisation of the distance between the actual joint pdf of the quantum system dynamics and an external electric field, and a predefined desired joint pdf, we provided the solution to the control problem for general quantum systems that are described by arbitrary pdfs. The distance between joint pdfs played the role of the cost function needed to be minimised and  is measured by the Kullback-Leibler divergence.  Furthermore, we provided the solution of the control problem for a particular class of quantum physical systems that are described by bilinear equations and Gaussian pdfs which depends on the parameters of the physical systems. There, it is shown that the obtained randomised controller is a feedback controller that can be calculated analytically. 

Our approach is fully probabilistic and is different to the current state of the art of quantum control which is mainly based on using deterministic objective functions. Additionally, the proposed approach is more accurate since it takes into consideration the quantum systems noises and uncertainties in the derivation of the optimal control law. 

To show the effectiveness of the proposed method we particularly applied it to control the outcome of particular target operators for a real molecule interacting with an external electric field and for spin systems. By analysing different experiments, we showed that our algorithm is efficient and in contrast to previous iterative methods the expected value of the target operator does not take a lot of time to converge to its desired value. 

A possible extension of this work is to apply the algorithm to consider the effect of the environment and  control the decoherence  occurring when a quantum system interacts with a large environment, which is one of the main challenges of the development of quantum technologies.  
\section*{Acknowledgements:} This work was supported by the EPSRC grant EP/V048074/1.
\section*{Appendices}

\appendix
\renewcommand{\theequation}{A.\arabic{equation}}
\setcounter{equation}{0} 
 
\section{Vectorisation of the density operator}\label{vect_app}
In this section we provide the matrices $\tilde{A}$ and $\tilde{N}$ appearing in equation (\ref{NLVN_pa}) for a tow-dimensional system. The vectorisation (\ref{x_t_vet}) of the density operator $\rho(t)	=\left(
\begin{array}{cc}
	\rho_{00}(t)&\rho_{01}(t)\\
	\rho_{10}(t) & \rho_{11}(t) \\
\end{array}\right)$ is given now by, 
\begin{equation}
\tilde{	x}(t)=\tilde{	x}_t=\left(
	\begin{array}{cc}
		\rho_{00}(t)\\\rho_{11}(t)\\
		\rho_{01} (t)\\ \rho_{10} (t)\\
	\end{array}\right),
\end{equation}
From equation (\ref{rho_elements}), we can show that now the matrices $\tilde{A}$ and $\tilde{N}$ are given as follows 
\begin{equation}
	\tilde{A}=\left(
	\begin{array}{cccc}
		-\gamma_{0,0} &\Gamma_{1\to 0}& 0 & 0\\
		\Gamma_{0\to 1} & -\gamma_{1,1} & 0& 0\\
		0 & 0 & -i\omega_{0,1}-\gamma_{0,1}& 0 \\
		0 & 0 & 0& -i\omega_{1,0}-\gamma_{1,0}\\
	\end{array}\right)
\end{equation}
and 
\begin{equation}
	\tilde{N}=\dfrac{1}{\hbar}\left(
	\begin{array}{cccc}
		0 &0& -\mu_{1,0} & \mu_{0,1}\\
		0 & 0 & \mu_{1,0}& -\mu_{0,1}\\
		-\mu_{0,1}& \mu_{0,1} & \mu_{0,0}-\mu_{1,1} & 0 \\
		\mu_{1,0}& -\mu_{1,0} & 0& \mu_{1,1} -\mu_{0,0}\\
	\end{array}\right).
\end{equation}
The results can be generalised to any $l$-dimensional physical system in a straightforward way. 
	\section{Complex normal distribution}\label{AppA_0}
	\setcounter{equation}{0} 
\renewcommand{\theequation}{B.\arabic{equation}}
The complex normal distribution for a complex random variable $x_t\in \mathbf{C}^n$ is given by,
\begin{align}\label{norcomplx}
	&\mathcal{N}_{\mathcal{C}}(\mu_{t},{\Gamma}^{'},C^{'})=\dfrac{1}{\pi^n\sqrt{\det(\Gamma^{'})\det(\bar{\Gamma}^{'}-(C^{'})^\dagger(\Gamma^{'})^{-1}C^{'})}} \exp\bigg[ \dfrac{1}{|\Gamma^{'}|^2-|C^{'}|^2}\bigg(- 0.5{({{x}_t}- {{\mu} _t})^{\dagger}}{\bar{\Gamma}^{'}}({x_t} - {\mu _t})\nonumber\\& + 0.5{({x_t} - {\mu _t})^T}{\bar{C}^{'}}({x}_t - {{\mu} _t}) + 0.5{({{x}_t} - {{\mu} _t})^{\dagger}}{C^{'}}({\bar{x}_t} - {\bar{\mu} _t}) - 0.5{({x_t} - {\mu _t})^T}{\Gamma^{'}}({\bar{x}_t} - {\bar{\mu} _t})\bigg)\bigg],
\end{align}
where, 
\begin{align}
\label{eq:relation}
x_t^\dagger&=\bar{x}^T,\nonumber\\
\mu_t& =E(x_t),\nonumber\\ 
\Gamma^{'}&= E((x_t-\mu_t)(x_t-\mu_t)^\dagger),\nonumber\\ 
C^{'}&= E((x_t-\mu_t)(x_t-\mu_t)^T).
\end{align}
Here  $E(.)$ stands for the expected value, $x_t^T$ is the matrix transpose of $x_t$, and $x_t^\dagger$ is the conjugate transpose of $x_t$. The matrices $\mu_t$, $\Gamma^{'}$ and $C^{'}$ are respectively the mean, the covariance and the relation or pseudo-covariance matrices.
\section{Calculation of the performance index $\gamma(x_{t-1})$}\label{AppA}
\setcounter{equation}{0} 
\renewcommand{\theequation}{C.\arabic{equation}}
The calculation of the From of the performance index function as given in Eq. (\ref{49}) requires the evaluation of the coefficient $\beta(u_{t-1},x_{t-1})$ defined in equation (\ref{beta}) and repeated here,
	\begin{align}\label{betaapp1}
		\beta(u_{t-1},x_{t-1}) = \int s(x_{t}|u_{t-1}, x_{t-1}) s(o_t| x_{t}) \ln \bigg ( \frac{ s(o_t| x_{t})}{^Is(o_t| x_{t})} \dfrac{1}{ {\gamma}(x_{t})}\bigg) \mathrm{d} x_{t}\mathrm{d} o_t.
	\end{align}
Let us first calculate  $\ln \bigg ( \frac{ s(o_t| x_{t})}{^Is(o_t| x_{t})} \dfrac{1}{ {\gamma}(x_{t})}\bigg)$. From equations (\ref{eq:eq5}),  (\ref{eq:eq2}) and (\ref{49}), we have,
	\begin{align}\label{betaapp2}
		& \ln \bigg ( \frac{ s(o_t| x_{t})}{^Is(o_t| x_{t})} \dfrac{1}{ {\gamma}(x_{t})}\bigg)\nonumber\\
		& 
		=- 0.5{({o_t} - D{x_t})^T}{G ^{ - 1}}({o_t} - D{x_t})+ 0.5{({o_t} - {{o}_{d}})^T}G_r^{ - 1}({o_t} - {{o}_{d}}) +0.5x_{t}^T{M_{t}}{x_{t}} +0.5P_{t}{x_{t}} + 0.5{\omega_{t}}+0.5\ln(\dfrac{|G_r|}{|G|}),	\end{align}
	where  $|G|$ stands for the determinant of the matrix $G$.
 This implies that,
\begin{align}
		& \ln \bigg ( \frac{ s(o_t| x_{t})}{^Is(o_t| x_{t})} \dfrac{1}{ {\gamma}(x_{t})}\bigg)\nonumber\\
& =  - 0.5o_t^T{G ^{ - 1}}{o_t} - 0.5x_t^T{D^T}{G ^{ - 1}}D{x_t} + o_t^T{G ^{ - 1}}D{x_t}  + 0.5o_t^{T} G_r^{ - 1}{o_t}- o_t^{T} G_r^{ - 1}{{o}_{d}} +0.5x_{t}^{T} {M_{t}}{x_{t}} +0.5P_{t}{x_{t}} +0.5 {\omega_{t}}\nonumber\\
		& + 0.5{o}_{d}^{T} G_r^{ - 1}{{o}_{d}}+0.5\ln(\dfrac{|G_r|}{|G|})  \nonumber\\
		&=  - 0.5o_t^{T} ({G ^{ - 1}} - G_r^{ - 1}){o_t} + o_t^{T} ({G ^{ - 1}}D{x_t} - G_r^{ - 1}{{o}_{d}}) -0.5 x_t^{T} D^{T} {G ^{ - 1}}D{x_t} +0.5x_{t}^{T} {M_{t}}{x_{t}} +0.5P_{t}{x_{t}} + 0.5{\omega_{t}} \nonumber\\
		&+ 0.5{o}_{d}^{T} G_r^{ - 1}{{o}_{d}}+0.5\ln(\dfrac{|G_r|}{|G|}).
	\end{align}
By substituting it back and integrate over  $o_t$, we find,
	\begin{align}\label{betaapp3}
		&\beta(u_{t-1},x_{t-1})=\int s(x_{t}|u_{t-1}, x_{t-1}) s(o_t| x_{t})\bigg( - 0.5o_t^{T} ({G ^{ - 1}} - G_r^{ - 1}){o_t} + o_t^{T} ({G ^{ - 1}}D{x_t} - G_r^{ - 1}{{o}_{d}}) -0.5 x_t^{T} D^{T} {G ^{ - 1}}D{x_t}\nonumber\\
		&  +0.5x_{t}^{T} {M_{t}}{x_{t}}+0.5P_{t}{x_{t}} + 0.5{\omega_{t}}+ 0.5{o}_{d}^{T} G_r^{ - 1}{{o}_{d}} +0.5\ln(\dfrac{|G_r|}{|G|})\bigg) \mathrm{d} x_{t}\mathrm{d} o_t\nonumber\\
		& = \int {s\left( {{x_t}|{u_{t-1}},{\rm{ }}{x_{t - 1}}} \right){\rm{ }}} \bigg( 0.5x_t^{T} {D^{T} }G_r^{ - 1}D{x_t} - x_t^{T} {D^{T} }G_r^{ - 1}{{o}_{d}}+0.5x_{t}^{T} {M_{t}}{x_{t}} +0.5P_{t}{x_{t}} + 0.5{\omega_{t}+}0.5{o}_{d}^{T} G_r^{ - 1}{{o}_{d}} \nonumber\\
		&-0.5 \Tr(G(G^{-1}-G_r^{ - 1})) +0.5\ln(\dfrac{|G_r|}{|G|})\bigg)d{x_t} \nonumber\\
		& = \int {s\left( {{x_t}|{u_{t-1}},{\rm{ }}{x_{t - 1}}} \right){\rm{ }}} \bigg( 0.5x_t^{T} ({D^{T} }G_r^{ - 1}D+M_{t}){x_t}+ 0.5(P_{t}-2{o}_{d}^{T} G_r^{ - 1}{D})x_t +0.5{\omega_{t}} + 0.5{o}_{d}^{T} G_r^{ - 1}{{o}_{d}} +0.5\ln(\dfrac{|G_r|}{|G|})  \nonumber\\
		&-0.5 \Tr(G(G^{-1}-G_r^{ - 1}))\bigg)d{x_t}
	\end{align}
	Now, we integrate over  $x_t$, we find,
	\begin{align} \label{betaapp4}
		&\beta(u_{t-1},x_{t-1})= 0.5\mu_t^{T} ({D^{T} }G_r^{ - 1}D+M_{t}){\mu_t} +  0.5(P_{t}-2{o}_{d}^{T} G_r^{ - 1}{D})\mu_t+ 0.5{\omega_{t}} + 0.5{o}_{d}^{T} G_r^{ - 1}{{o}_{d}}  +0.5\ln(\dfrac{|G_r|}{|G|})\nonumber\\
		&-0.5 \Tr(G(G^{-1}-G_r^{ - 1}))-0.5\Tr(\bar{C}^{-1}({D^{T} }G_r^{ - 1}D+M_{t})),
	\end{align}
where $C=\dfrac{C^{'}}{|\Gamma^{'}|^2-|C^{'}|^2}$. 
By replacing $\mu_t=Ax_{t-1}+Bu_{t-1}$, we find,
	\begin{align}\label{betaapp5}
	&\beta(u_{t-1},x_{t-1})= 0.5(Ax_{t-1}+Bu_{t-1})^{T} ({D^{T} }G_r^{ - 1}D+M_{t}){(Ax_{t-1}+Bu_{t-1})} +0.5(P_{t}-2{{o}_{d}^{T} }G_r^{ - 1}D)  (Ax_{t-1}+Bu_{t-1}) \nonumber\\&
+	0.5{\omega_{t}} + 0.5{o}_{d}^{T} G_r^{ - 1}{{o}_{d}}  +0.5\ln(\dfrac{|G_r|}{|G|})-0.5 \Tr(G(G^{-1}-G_r^{ - 1}))-0.5\Tr(\bar{C}^{-1}({D^{T} }G_r^{ - 1}D+M_{t}))
	\nonumber\\	&= 0.5x_{t-1}^{T} A^{T} ({D^{T} }G_r^{ - 1}D+M_{t}){Ax_{t-1}} +0.5(P_{t}-2{{o}_{d}^{T} }G_r^{ - 1}D) Ax_{t-1} + 0.5 u_{t-1}^{T} B^{T}({D^{T} }G_r^{ - 1}D+M_{t}){Bu_{t-1}} \nonumber\\
	&+ u_{t-1}^{T} B^{T}\big(({D^{T} }G_r^{ - 1}D+M_{t}){Ax_{t-1}}+0.5(P_{t}-2{{o}_{d}^{T} }G_r^{ - 1}D)^{T} \big) + 
	0.5{\omega_{t}} + 0.5{o}_{d}^{T} G_r^{ - 1}{{o}_{d}}  +0.5\ln(\dfrac{|G_r|}{|G|})\nonumber\\&-0.5 \Tr(G(G^{-1}-G_r^{ - 1}))-0.5\Tr(\bar{C}^{-1}({D^{T} }G_r^{ - 1}D+M_{t}))
\end{align}
By projecting the form of $\beta$ found in (\ref{betaapp5}) and the ideal distribution of the controller provided in equation (\ref{eq:eq3}), in the definition of $\gamma(x_{t-1})$ given in equation (\ref{gamma2}), we find that,
	\begin{align}\label{gamma}
		&\gamma(x_{t-1}) = \int {^Ic(u_{t-1}|x_{t-1}) \exp[{-\beta(u_{t-1},x_{t-1})}] \mathrm{d} u_{t-1}}, \nonumber \\
		& =(2\pi )^{-1/2}|\Omega|^{-1/2}\int {\exp } \bigg[ - 0.5{({u_{t-1}} - {u_r})^{T} }{\Omega ^{ - 1}}({u_{t-1}} - {u_r}) -0.5x_{t-1}^{T} A^{T} ({D^{T} }G_r^{ - 1}D+M_{t}){Ax_{t-1}} -0.5(P_{t}-2{{o}_{d}^{T} }G_r^{ - 1}D) \nonumber\\
		&Ax_{t-1} - 0.5u_{t-1}^{T} B^{T} ({D^{T} }G_r^{ - 1}D+M_{t}){Bu_{t-1}}-u_{t-1}^{T} B^{T} \big(({D^{T} }G_r^{ - 1}D+M_{t}){Ax_{t-1}}+0.5(P_{t}-2{{o}_{d}^{T} }G_r^{ - 1}D)^{T} \big)  - 0.5{\omega_{t}}\nonumber\\& - 0.5{o}_{d}^{T} G_r^{ - 1}{{o}_{d}}  -0.5\ln(\dfrac{|G_r|}{|G|})+0.5 \Tr(G(G^{-1}-G_r^{ - 1}))+0.5\Tr(\bar{C}^{-1}({D^{T} }G_r^{ - 1}D+M_{t}))\bigg] {\rm{d}}{u_{t-1}} \nonumber\\
		&=(2\pi)^{-1/2}|\Omega|^{-1/2} {\exp } \bigg(-0.5x_{t-1}^{T} A^{T} ({D^{T} }G_r^{ - 1}D+M_{t}){Ax_{t-1}} -0.5(P_{t}-2{{o}_{d}^{T} }G_r^{ - 1}D) Ax_{t-1} 
		- 0.5{\omega_{t}}- 0.5{o}_{d}^{T} G_r^{ - 1}{{o}_{d}}  \nonumber\\& -0.5\ln(\dfrac{|G_r|}{|G|})+0.5 \Tr(G(G^{-1}-G_r^{ - 1}))+0.5\Tr(\bar{C}^{-1}({D^{T} }G_r^{ - 1}D+M_{t}))
			-0.5u_r^{T} \Omega^{-1}u_r\bigg)\nonumber\\
		&\int {\exp } \bigg(-0.5u_{t-1}^{T} (\Omega^{-1}+B^{T} ({D^{T} }G_r^{ - 1}D+M_{t})B){u_{t-1}}+u_{t-1}^{T} \big(\Omega^{-1}u_r-B^{T} ({D^{T} }G_r^{ - 1}D+M_{t}){Ax_{t-1}}\nonumber\\
		&-0.5B^{T} (P_{t}-2{{o}_{d}^{T} }G_r^{ - 1}D)^{T} \big) \bigg) {\rm{d}}{u_{t-1}} 
			\end{align} 
The integral in equation (\ref{gamma}) is a particular case of the general multiple integral given in Theorem (10.5.1) in \cite{raybill}. Hence, it follows that,
	\begin{align}\label{gamma_1}
&\gamma(x_{t-1})= {\exp } \bigg(-0.5x_{t-1}^{T} A^{T} ({D^{T} }G_r^{ - 1}D+M_{t}){Ax_{t-1}} -0.5(P_{t}-2{{o}_{d}^{T} }G_r^{ - 1}D) Ax_{t-1} - 0.5{\omega_{t}}- 0.5{o}_{d}^{T} G_r^{ - 1}{{o}_{d}}  \nonumber\\& -0.5\ln(\dfrac{|G_r|}{|G|})+0.5 \Tr(G(G^{-1}-G_r^{ - 1}))+0.5\Tr(\bar{C}^{-1}({D^{T} }G_r^{ - 1}D+M_{t}))
-0.5u_r^{T} \Omega^{-1}u_r\bigg) \nonumber\\
&{\exp }\bigg(0.5\big(\Omega^{-1}u_r-B^{T} ({D^{T} }G_r^{ - 1}D+M_{t}){Ax_{t-1}}-0.5B^{T} (P_{t}-2{{o}_{d}^{T} }G_r^{ - 1}D)^{T} \big)^{T} 
\big(\Omega^{-1}+B^{T} ({D^{T} }G_r^{ - 1}D+M_{t})B\big)^{-1}\nonumber\\
&
\big(\Omega^{-1}u_r-B^{T} ({D^{T} }G_r^{ - 1}D+M_{t}){Ax_{t-1}}-0.5B^{T} (P_{t}-2{{o}_{d}^{T} }G_r^{ - 1}D)^{T} \big)
 \bigg)\times{|\Omega|^{-1/2}}|\Omega^{-1}+B^{T} ({D^{T} }G_r^{ - 1}D+M_{t})B|^{-1/2}
	\end{align}
Finally,
	\begin{align}\label{appendb2}
		&-\ln{(\gamma(x_{t-1}))}\nonumber\\
		&=  0.5x_{t-1}^{T} \bigg(A^{T} (D^{T} G_r^{-1}D+M_{t})A-A^{T} (D^{T} G_r^{-1}D+M_{t})^{T} B
		\big(\Omega^{-1}+B^{T} ({D^{T} }G_r^{ - 1}D+M_{t})B\big)^{-1}
		\nonumber\\
	&	B^{T} (D^{T} G_r^{-1}D+M_{t})A\bigg)x_{t-1}
	\nonumber\\
&	+0.5\bigg((P_{t}-2{o}_{d}^{T} G_r^{-1}D)A+2(\Omega^{-1}u_r-0.5B^{T} (P_{t}^{T} -2D^{T} G_r^{-1}{o}_{d}))^{T} 
		\big(\Omega^{-1}+B^{T} ({D^{T} }G_r^{ - 1}D+M_{t})B\big)^{-1}\nonumber\\
		&B^{T} (D^{T} G_r^{-1}D+M_{t})A
	\bigg)	x_{t-1}	\nonumber\\
	&+0.5\bigg({\omega_{t}}+ {o}_{d}^{T} G_r^{ - 1}{{o}_{d}}   +\ln(\dfrac{|G_r|}{|G|})- \Tr\big(G(G^{-1}-G_r^{ - 1})\big)-\Tr(\bar{C}^{-1}({D^{T} }G_r^{ - 1}D+M_{t}))
	+u_r^{T} \Omega^{-1}u_r	\nonumber\\
	&-\big( \Omega^{-1}u_r-0.5B^{T} (P_{t}^{T} -2{D}^{T}G_r^{ - 1}{{o}_{d}})\big)^{T} 
	\big(\Omega^{-1}+B^{T} ({D^{T} }G_r^{ - 1}D+M_{t})B\big)^{-1}	
	\big( \Omega^{-1}u_r-0.5B^{T} (P_{t}^{T} -2{D}^{T}G_r^{ - 1}{{o}_{d}})\big)\nonumber\\
	&-2\ln({|\Omega|^{-1/2}}|\Omega^{-1}+B^{T} ({D^{T} }G_r^{ - 1}D+M_{t})B|^{-1/2})\bigg).
	\end{align}
\section{Calculation of the control distribution function}\label{appC}
\setcounter{equation}{0} 
\renewcommand{\theequation}{D.\arabic{equation}}
The form of the distribution of the optimal controller is provided in equation (\ref{eq:Eqn4}),
	\begin{equation}
	c^{*}(u_{t-1}|x_{t-1})=\frac{^Ic(u_{t-1}|x_{t-1}) \exp [-\beta(u_{t-1},x_{t-1})]}{\gamma(x_{t-1})}.
\end{equation}
From equations (\ref{eq:eq3}), (\ref{betaapp5}) and (\ref{appendb2}), we find 
\begin{align}
 &c^{*}(u_{t-1}|x_{t-1})=(2\pi)^{-1/2}|\Omega^{-1}+B^{T} ({D^{T} }G_r^{ - 1}D+M_{t})B|^{1/2}\exp\bigg[-0.5u_{t-1}^{T}\bigg(\Omega^{-1}+B^{T} ({D^{T} }G_r^{ - 1}D+M_{t})B\bigg){u_{t-1}}\nonumber\\
 &+u_{t-1}^{T} \bigg(\Omega^{-1}u_r-B^{T} ({D^{T} }G_r^{ - 1}D+M_{t}){Ax_{t-1}}
 -0.5B^{T} (P_{t}-2{{o}_{d}^{T} }G_r^{ - 1}D)^{T} \bigg) 
-0.5\bigg({\Omega ^{ - 1}}{u_r} - {B^{T} }({D^{T} }G_r^{ - 1}D + {M_t})A{x_{t -1}}\nonumber\\
& -0.5 {B^{T} }( P_{t }^{T}  - 2{D^{T} }G_r^{ - 1}{{o}_{d}} )\bigg)^T\bigg({\Omega ^{ - 1}} + {B^{T} }( {D^{T} }G_r^{ - 1}D + {M_t})B\bigg)^{-1} \bigg({\Omega ^{ - 1}}{u_r} - {B^{T} }({D^{T} }G_r^{ - 1}D + {M_t})A{x_{t -1}}\nonumber\\
& -0.5 {B^{T} }( P_{t }^{T}  - 2{D^{T} }G_r^{ - 1}{{o}_{d}} )\bigg)
\bigg].
\end{align}
After simplification we find that, 
\begin{equation}
c^{*}(u_{t-1}|x_{t-1})\sim \mathcal{N}(v_{t-1},R_{t}).
\end{equation}
where, 
\begin{align}\label{optimalcappc}
	v_{t-1}&=\bigg({\Omega ^{ - 1}} + {B^{T} }( {D^{T} }G_r^{ - 1}D + {M_t})B\bigg)^{-1} \bigg({\Omega ^{ - 1}}{u_r} - {B^{T} }({D^{T} }G_r^{ - 1}D + {M_t})A{x_{t -1}} -0.5 {B^{T} }( P_{t }^{T}  - 2{D^{T} }G_r^{ - 1}{{o}_{d}} )\bigg),
\end{align}
and,
\begin{align}
	R_t&=\bigg({\Omega ^{ - 1}} + {B^{T} }( {D^{T} }G_r^{ - 1}D + {M_t})B\bigg)^{-1}.
\end{align}
	\section{State space model for spin $\frac{1}{2}$}\label{spinn_1_2}
	\setcounter{equation}{0} 
\renewcommand{\theequation}{E.\arabic{equation}}
	The evolution of a closed spin-1/2 system in interaction with an external electric field can be described by the following master equation,
	\begin{equation}\label{vonspin}
		\dfrac{d\rho(t)}{dt}=-i[H,\rho(t)],
	\end{equation}
where, as given in equation (\ref{58_}), $H=\dfrac{1}{2}\sigma_3+\dfrac{1}{2}(\sigma_1+\sigma_2) u(t)$ such that $u(t)$ is the electric field and $\sigma_1,\sigma_2, \sigma_3$ are the Pauli matrices  whose values are given in equation (\ref{PauliMat}), repeated here, 
	\begin{equation}
		\sigma_1	=\left(
		\begin{array}{cc}
			0&1\\
			1 & 0 \\
		\end{array}
		\right),\hspace{0.4cm}	\sigma_2=\left(
		\begin{array}{cc}
			0 &-i\\
			i & 0 \\
		\end{array}
		\right),\hspace{0.4cm}\sigma_3	=\left(
		\begin{array}{cc}
			1 &0\\
			0 & -1\\
		\end{array}
		\right).
	\end{equation}
	Consequently, the Hamiltonian, $H$ can be evaluated to give, 
	\begin{equation}
		H=\dfrac{1}{2}\left(
		\begin{array}{cc}
			1 & u(t)(1-i) \\
			u(t)(1+i) & -1 \\
		\end{array}
		\right).
	\end{equation} 
Considering the form of the density operator, 
	\begin{equation}
		\rho(t)=\left(
		\begin{array}{cc}
			\rho_{00}(t) & \rho_{01}(t)\\
			\rho^*_{01}(t)  & \rho_{11}(t)
		\end{array}
		\right),
	\end{equation} 
the Liouville von-Neumann equation (\ref{vonspin}) can be re-written as,
	\begin{equation}
		\dfrac{d}{dt}\left(
		\begin{array}{cc}
			\rho_{00}(t) & \rho_{01}(t)\\
		\rho^*_{01}(t)  &\rho_{11}(t)
		\end{array}
		\right)=\dfrac{-i}{2}\left(
		\begin{array}{cc}
			u(t)\big((1-i) \rho^*_{01}(t)-(1+i) \rho_{01}(t)\big) &2\rho_{01}(t)+u(t)(1-i)(\rho_{11}(t)-\rho_{00}(t))\\
			-2\rho^*_{01}(t) +u(t)(1+i)(\rho_{00}(t) -\rho_{11}(t) ) & u(t)\big((1+i)\rho_{01}(t) -(1-i)\rho^*_{01}(t) \big)  \\
		\end{array}
		\right).
	\end{equation}
where the elements $\rho_{00}(t), \rho_{01}(t), \rho^*_{01}(t)$ and $\rho_{11}(t)$ are the elements of the density operator, $\rho(t)$. Vectorising the above equation as discussed in (\ref{x_t_vet}) yields,
	\begin{align}
	&	\dfrac{d}{dt}\underbrace{\left(
		\begin{array}{cc}
			\rho_{00}(t)\\ \rho_{11}(t)\\
			\rho_{01}(t) \\ \rho^*_{01}(t)
		\end{array}
		\right)}_{x(t)}= \underbrace{\left(
		\begin{array}{cccc}
			0 & 0 & 0 & 0\\
			0 & 0 & 0& 0\\
			0 & 0 & -i& 0 \\
			0 & 0 & 0& i\\
		\end{array}\right)}_{\tilde{A}} \underbrace{\left(
		\begin{array}{cc}
			\rho_{00}(t)\\ \rho_{11}(t)\\
			\rho_{01}(t) \\ \rho^*_{01}(t)
		\end{array}
		\right)}_{x(t)}\\&+i \underbrace{\dfrac{ 1}{2}\left(
		\begin{array}{cccc}
			0 & 0 & (1+i) & -(1-i)\\
			0 & 0 & -(1+i) & (1-i)\\
			(1-i) & -(1-i)&0 & 0 \\
			-(1+i) & (1+i)&0 & 0 \\
		\end{array}\right)}_{\tilde{N}} \underbrace{\left(
	\begin{array}{cc}
		\rho_{00}(t)\\ \rho_{11}(t)\\
			\rho_{01}(t) \\ \rho^*_{01}(t)
	\end{array}
		\right)}_{x(t)}u(t),
	\end{align}
which yields the state equation for the spin-1/2 system in the form given in equation (\ref{NLVN_pa}), repeated here,
	\begin{equation}
	\dfrac{d{x}(t)}{d t} =(\tilde{A}+i\tilde{N} u(t) ){x}(t).
	\end{equation}
	\section{State space model for spin-1}\label{spinn_1}
	\setcounter{equation}{0} 
\renewcommand{\theequation}{F.\arabic{equation}}
	The Hamiltonian describing the interaction between a spin-1 system and an electric field $u(t)$ is given by, 
	\begin{equation}
		H=H_0+H_u(t)=\left(\begin{array}{ccc}
			\dfrac{3}{2} & 0 & 0  \\
			0 & 1 & 0  \\
			0 & 0&0   \\
		\end{array}\right)+\left(\begin{array}{ccc}
			0 & 0 & 1  \\
			0 & 0 & 1  \\
			1 & 1&0   \\
		\end{array}\right) u(t)=\left(\begin{array}{ccc}
			\dfrac{3}{2} & 0 & u(t)  \\
			0 & 1& u(t)  \\
			u(t) & u(t)&0   \\
		\end{array}\right)
	\end{equation} 
Considering the form of the density operator, 
	\begin{equation}
		\rho=\left(\begin{array}{ccc}
			\rho_{00}(t) & 	\rho_{01}(t) & 	\rho_{02}(t)\  \\
			\rho^*_{01}(t) & 	\rho_{11}(t) & \rho_{12}(t)\ \\
			\rho^*_{02}(t) &  \rho^*_{12}(t) &	\rho_{22}(t)
		\end{array}\right),
	\end{equation}
 the master equation $	\dfrac{d\rho(t)}{dt}=-i[H,\rho(t)]$ can be written as, 
{\small	\begin{align}
	&	i\dfrac{d}{dt}\left(\begin{array}{ccc}
			\rho_{00}(t) & 	\rho_{01}(t) & 	\rho_{02}(t)\  \\
			\rho^*_{01}(t) & 	\rho_{11}(t) & \rho_{12}(t)\ \\
			\rho^*_{02}(t) &  \rho^*_{12}(t) &	\rho_{22}(t)
		\end{array}\right)=\\&\left(\begin{array}{ccc}
			u(t)\big(\rho^*_{02}(t) -\rho_{02}(t)\big)& \dfrac{1}{2}\rho_{01}(t)+u(t)\big( \rho^*_{12}(t)-\rho_{02}(t)\big ) & \dfrac{3}{2}\rho_{02}(t) +u(t)\big(\rho_{22}(t)-\rho_{00}(t) -\rho_{01}(t)\big)  \\\\
			-\dfrac{1}{2} \rho^*_{01}(t)+u(t)(\rho^*_{02}(t)- \rho_{12}(t))& u(t)\big( \rho_{12}^*(t)- \rho_{12}(t)\big)& \rho_{12}(t)+u(t)\big(\rho_{22}(t) -\rho_{01}^*(t)-\rho_{11}(t)\big) \\\\
			-\dfrac{3}{2} \rho_{02}^*(t) +u(t)\big(\rho_{00} +\rho_{01}^*(t)-\rho_{22}(t)\big)& - \rho_{12}^*(t)+u(t)(\rho_{01}(t)+\rho_{11}(t)-\rho_{22} (t)) & u(t)\big(\rho_{02}(t) + \rho_{12}(t)-\rho_{02}^*(t) - \rho_{12}^*(t)\big) \\
		\end{array}\right).
	\end{align}}
This implies that,
\newpage
	\begin{eqnarray}
		&\dfrac{d}{dt}\underbrace{\left(
		\begin{array}{ccccccccc}
			\rho_{00} (t) \\
		\rho_{11} (t)\\
			\rho_{22}(t) \\
			\rho_{01}(t) \\
				\rho_{02}(t)  \\
			\rho_{01} ^* (t) \\
			\rho_{02}^*(t)  \\
				\rho_{12}(t)\\
			\rho_{12}^* (t) \\
		\end{array}\right)}_{x(t)}=\underbrace{\left(
		\begin{array}{ccccccccc}
			0 & 0 & 0 & 0 & 0 & 0 & 0 & 0 & 0 \\
			0 & 0 & 0 & 0 & 0 & 0 & 0 & 0 & 0 \\
			0 & 0 & 0 & 0 & 0 & 0 & 0 & 0 & 0 \\
			0 & 0 & 0 & -\dfrac{i}{2} & 0 & 0 & 0 & 0 & 0 \\
			0 & 0 & 0 & 0 & -\dfrac{3i}{2}& 0 & 0 & 0 & 0 \\
			0 & 0 & 0 & 0 & 0 & \dfrac{i}{2} & 0 & 0 & 0 \\
			0 & 0 & 0 & 0 & 0 & 0 & \dfrac{3i}{2} & 0 & 0 \\
			0 & 0 & 0 & 0 & 0 & 0 & 0 & -i& 0 \\
			0 & 0 & 0 & 0 & 0 & 0 & 0 & 0 & i\\
		\end{array}
		\right)}_{\tilde{A}} \underbrace{\left(
		\begin{array}{ccccccccc}
			\rho_{00} (t) \\
		\rho_{11} (t)\\
			\rho_{22}(t) \\
			\rho_{01}(t) \\
				\rho_{02}(t)  \\
			\rho_{01} ^* (t) \\
			\rho_{02}^*(t)  \\
				\rho_{12}(t)\\
			\rho_{12}^* (t) \\
		\end{array}\right)}_{x(t)} \nonumber \\
&+i\underbrace{\left(
		\begin{array}{ccccccccc}
			0 & 0 & 0 & 0 & 1 & 0 & -1 & 0 & 0 \\
			0 & 0 & 0 & 0 & 0 & 0 & 0 & 1 & -1 \\
			0 & 0 & 0 & 0 & -1& 0 & 1 & -1 & 1 \\
			0 & 0 & 0 & 0 & 1 & 0 & 0 & 0 & -1 \\
			1 & 0 & -1 & 1 & 0 & 0 & 0 & 0 & 0 \\
			0 & 0 & 0 & 0 & 0 & 0 & -1 & 1 & 0 \\
			-1 & 0 & 1& 0 & 0 & -1 & 0 & 0 & 0 \\
			0 & 1 & -1 & 0 & 0 & 1 & 0 & 0 & 0 \\
			0 & -1 & 1 & -1 & 0 & 0 & 0 & 0 & 0 \\
		\end{array}
		\right)}_{\tilde{N}}\underbrace{\left(
		\begin{array}{ccccccccc}
			\rho_{00} (t) \\
		\rho_{11} (t)\\
			\rho_{22}(t) \\
			\rho_{01}(t) \\
				\rho_{02}(t)  \\
			\rho_{01} ^* (t) \\
			\rho_{02}^*(t)  \\
				\rho_{12}(t)\\
			\rho_{12}^* (t) \\
		\end{array}\right)}_{x(t)} u(t).
	\end{eqnarray}
Hence, the time evolution of the vectorised density operator can be written in the form given in (\ref{NLVN_pa}) as follows,
	\begin{equation}
	\dfrac{d{x}(t)}{d t} =(\tilde{A}+i\tilde{N} u(t) ){x}(t).
	\end{equation}

\newpage{\pagestyle{empty}\cleardoublepage}
	
\end{document}